\documentclass[12pt]{article}
\usepackage[margin=1in]{geometry}
\usepackage[T1]{fontenc}
\usepackage[utf8]{inputenc}
\usepackage[authoryear,longnamesfirst]{natbib}
\usepackage{amsmath,amssymb,amsthm}
\usepackage{mathtools}
\usepackage{bm}
\usepackage{dsfont}
\usepackage{graphicx}
\usepackage{booktabs}
\usepackage[table]{xcolor}
\usepackage{threeparttable}
\usepackage{tabularx}
\usepackage{array}
\usepackage{enumitem}
\usepackage{needspace}
\usepackage{setspace}
\usepackage{etoolbox}
\usepackage{microtype}
\usepackage{caption}
\usepackage{longtable}
\usepackage{float}
\usepackage{xr-hyper}
\usepackage[colorlinks,
            linkcolor=blue, anchorcolor=blue, citecolor=blue,
            urlcolor=blue, filecolor=blue]{hyperref}
\usepackage[nameinlink,noabbrev]{cleveref}

\AtBeginEnvironment{table}{\setstretch{1}}
\AtBeginEnvironment{longtable}{\singlespacing}
\AtBeginEnvironment{thebibliography}{\singlespacing}
\hypersetup{
  pdfauthor={Peter Reinhard Hansen and Chen Tong},
  pdfsubject={A balanced parametrization of correlation matrices},
  pdfkeywords={correlation matrix, unconstrained parametrization,
    balanced correlation parametrization, asymptotic sparsity, high-dimensional correlation}
}
\newcommand{\R}{\mathbb{R}}
\newcommand{\one}{\iota}
\newcommand{\indic}{\mathds{1}}

\newcommand{\Sym}[1]{\mathbb{S}_{#1}}
\newcommand{\Symp}[1]{\mathbb{S}_{#1}^{++}}
\newcommand{\CorP}[1]{\mathcal{C}_{#1}}
\newcommand{\Hol}[1]{\mathbb{H}_{#1}}
\newcommand{\HolBlock}[1]{\mathbb{H}_{#1}^{\mathcal P}}
\newcommand{\CorBlock}[1]{\mathcal{C}_{#1}^{\mathcal P}}
\newcommand{\Bal}[1]{\mathbb{B}_{#1}}

\newcommand{\Diag}{\operatorname{Diag}}
\newcommand{\diag}{\operatorname{diag}}
\newcommand{\Off}{\operatorname{Off}}
\newcommand{\Cor}{\operatorname{Cor}}
\newcommand{\vecl}{\operatorname{vecl}}
\newcommand{\tr}{\operatorname{tr}}
\newcommand{\rank}{\operatorname{rank}}

\newcommand{\LogLS}{\operatorname{Log}_{\mathrm{LS}}}

\newcommand{\cinv}{\mathcal{I}}
\newcommand{\arsinh}{\operatorname{arsinh}}
\newcommand{\atanh}{\operatorname{atanh}}

\newcommand{\norm}[1]{\left\lVert #1\right\rVert}
\newcommand{\normF}[1]{\left\lVert #1\right\rVert_{\!F}}
\newcommand{\ip}[2]{\left\langle #1,#2\right\rangle}
\newcommand{\Had}{\mathbin{\circ}}
\newcommand{\dd}{\mathrm{d}}
\newcommand{\cO}{\mathcal{O}}

\newcommand{\sgnperm}{\mathfrak{S}^{\pm}_n}

\theoremstyle{plain}
\newtheorem{theorem}{Theorem}
\newtheorem{proposition}{Proposition}
\newtheorem{corollary}{Corollary}
\newtheorem{lemma}{Lemma}
\newtheorem{assumption}{Assumption}

\theoremstyle{definition}
\newtheorem{definition}{Definition}
\newtheorem{example}{Example}
\crefname{example}{example}{examples}
\Crefname{example}{Example}{Examples}

\theoremstyle{remark}
\newtheorem{remark}{Remark}

\makeatletter
\renewenvironment{proof}[1][\proofname]{\par
  \pushQED{\qed}%
  \normalfont \topsep6\p@\@plus6\p@\relax
  \trivlist
  \item[\hskip\labelsep\bfseries #1\@addpunct{.}]\ignorespaces
}{%
  \popQED\endtrivlist\@endpefalse
}
\makeatother

\crefname{theorem}{theorem}{theorems}
\Crefname{theorem}{Theorem}{Theorems}
\crefname{proposition}{proposition}{propositions}
\Crefname{proposition}{Proposition}{Propositions}
\crefname{corollary}{corollary}{corollaries}
\Crefname{corollary}{Corollary}{Corollaries}
\crefname{lemma}{lemma}{lemmas}
\Crefname{lemma}{Lemma}{Lemmas}
\crefname{assumption}{assumption}{assumptions}
\Crefname{assumption}{Assumption}{Assumptions}
\crefname{definition}{definition}{definitions}
\Crefname{definition}{Definition}{Definitions}
\crefname{remark}{remark}{remarks}
\Crefname{remark}{Remark}{Remarks}
\newcommand{\keeptheoremstart}[1]{%
  \AtBeginEnvironment{#1}{\Needspace{3\baselineskip}}}
\forcsvlist{\keeptheoremstart}{theorem,proposition,corollary,lemma,assumption,definition,remark,example}
\AtBeginEnvironment{proof}{\Needspace{3\baselineskip}}
\newcommand{\keywords}[1]{\par\smallskip\noindent\textbf{Keywords:} #1}
\newcommand{\msc}[1]{\par\smallskip\noindent\textbf{MSC 2020:} #1}

\author{Peter Reinhard Hansen$^{\mathsection}$ \quad and \quad Chen Tong$^{\ddagger}$\\[1ex]
  {\small $^{\mathsection}$Department of Economics, University of North Carolina at Chapel Hill}\\
  {\small $^{\ddagger}$School of Economics, Xiamen University}}
\date{October 2026}
\hypersetup{pdftitle={A Balanced Parametrization of Correlation Matrices}}

\title{\textbf{A Balanced Parametrization of Correlation Matrices}\thanks{Corresponding author: Peter Reinhard Hansen. Email: hansen@unc.edu. Chen Tong acknowledges financial support
from the National Natural Science
Foundation of China (72301227, 72673150) and the Fujian Provincial Natural Science
Foundation of China (2025J08008).
}}

\makeatletter
\let\BCPauthor\@author
\let\BCPdate\@date
\makeatother
\begin{document}
\maketitle

\onehalfspacing
 
\begin{abstract}
We introduce the \emph{balanced correlation parametrization} (BCP), which maps positive-definite correlation matrices one-to-one onto the set of symmetric hollow (zero-diagonal) matrices with unrestricted off-diagonal entries. A key feature of the parametrization is that both the transformation and its inverse are available in closed form. It can be viewed as a generalization of the transformation underlying the classical correlation $t$-statistic to higher dimensions. It has a direct interpretation as a scaled sum of marginal and partial correlations, is inverse-consistent and equivariant under variable permutations and sign changes, and preserves block structure. A novel property of BCP is that it preserves correlation sparsity as the dimension tends to infinity under suitable factor conditions, providing a basis for sparse high-dimensional correlation modeling. An empirical illustration using standardized idiosyncratic daily stock returns for 1,765 equities over 15 years shows that BCP retains the approximate sparsity pattern and magnitudes of shock correlations, and that it yields accurate positive-definite reconstructions under sparse coordinate restrictions.
\end{abstract}

\msc{15B48; 15A16; 15A24; 53B20; 62H20}

\newpage

\section{Introduction}

Correlation matrices are central to multivariate models in statistics,
econometrics, and finance. They enter likelihoods, copula models,
and the analysis of portfolio risk.\footnote{Examples include multivariate GARCH
models with dynamic correlations
\citep{Engle:2002,TseTsui:2002,CrealKoopmanLucas:2011,EngleKelly:2012,Aielli:2013,EngleLedoitWolf:2019,PakelShephardSheppardEngle:2021,ArchakovHansenLundeMRG,TongHansenArchakov:2026,TongHansen:2025FactorGARCH}
and dynamic copula models with time-varying correlations
\citep{CrealTsay2015,OhPatton:2018,OhPatton:2023,OpschoorLucasBarraVanDick:2021,HafnerWang:2023}.}
Positive definiteness constrains the individual correlation coefficients,
which prevents their use as unrestricted parameters.\footnote{The
constraints are severe in high dimensions, where valid correlation
matrices occupy a vanishing fraction of the cube of entries,
see \citet{BohmHornik:2014} and \citet{Hansen:CorrelationMatrix}.}
Covariance normalization, as in the DCC model of \citet{Engle:2002}, rescales
an auxiliary positive-definite matrix to unit diagonal. However, this matrix
is not uniquely identified, since positive diagonal rescalings leave its
correlations unchanged.

The literature has introduced several one-to-one parametrizations that
avoid both problems. These include Cholesky and modified Cholesky
factorizations
\citep{PinheiroBates:1996,Pourahmadi:1999,vanOest:2021,ThanwerdasPennec:2022},
spherical coordinates
\citep{RapisardaBrigoMercurio:2007,PourahmadiWang:2015,TsayPourahmadi:2017,ForresterZhang:2020,LucchettiPedini:2024},
and sequential or vine-based partial correlations
\citep{KurowickaCooke:2003,Joe:2006,DanielsPourahmadi:2009,LewandowskiKurowickaJoe:2009,SturmfelsTsukermanWilliams:2016}.
These constructions provide explicit, computationally convenient coordinates
that can be transformed componentwise into unrestricted real parameters.
However, the coordinates depend on the ordering of the variables.
Matrix-logarithmic covariance models avoid this dependence
\citep{LeonardHsu:1992,ChiuLeonardTsui:1996}. However, the matrix logarithms
of correlation matrices do not fill a linear space, because their diagonal
entries are constrained by the unit-diagonal normalization.

\citet{ArchakovHansen:Correlation} overcome this problem by retaining the
off-diagonal entries of the matrix logarithm. Their generalized Fisher
transformation (GFT) is a global parametrization that does not depend on variable ordering.
\citet{Zwiernik:2025} applies the same idea to the inverse correlation
matrix, whose off-diagonal entries give another unrestricted
parametrization (InvC). However, inversion of both GFT and InvC requires a
numerically determined diagonal correction. The log-scaling (LS)
transformation of \citet{Thanwerdas:2024} has a direct inverse and treats
normalized correlation inversion symmetrically, but its transformation from
correlation matrices to the LS coordinates is not in closed form.

In this paper, we introduce the balanced correlation parametrization (BCP), which
maps positive-definite correlation matrices one-to-one onto symmetric hollow
(zero-diagonal) matrices with unrestricted off-diagonal entries. This means
that positive-definite correlation matrices can be parametrized by an
unrestricted vector in $\R^{n(n-1)/2}$, as is the case for GFT, InvC, and LS.
For a positive-definite correlation matrix
$C$, the transformation $B(C)$ and its inverse are given by
\begin{equation*}
 \begin{aligned}
  B(C)&=\tfrac12\left(D^{1/4}CD^{1/4}-D^{-1/4}C^{-1}D^{-1/4}\right),\\
  C(B)&=\Cor ( B+\sqrt{I_n+B^2}),
 \end{aligned}
\end{equation*}
where $D=\Diag(C^{-1})$ is the diagonal part of $C^{-1}$, and
$\Cor(\Sigma)=\Diag(\Sigma)^{-1/2}\Sigma\Diag(\Sigma)^{-1/2}$
rescales a positive-definite symmetric matrix to have unit diagonal.
The key advantage of BCP over GFT, InvC, and LS is that both the
transformation and its inverse are available in closed form and require
no iterative scaling or diagonal completion. Moreover, BCP has a range of additional convenient properties, including:
explicit derivatives, which support estimation and likelihood optimization;
equivariance under variable permutations and sign changes; inverse
consistency; exact preservation of block structure; and interpretable
coordinates that combine marginal and partial correlations.
A novel property of BCP is that it preserves correlation sparsity as the
dimension tends to infinity under suitable factor conditions, providing
a basis for sparse high-dimensional correlation modeling.
The other parametrizations share some of these properties, as summarized
in \Cref{tab:method-comparison}.

In the bivariate case, applying BCP to the sample correlation yields
exactly the classical $t$-statistic up to a sample-size scaling factor. Let $\hat{\rho}$ denote
the sample correlation from $T$ i.i.d. bivariate Gaussian observations.
Under the null hypothesis of zero correlation,
\begin{equation*}
 \frac{\hat{\rho}}{\sqrt{1-\hat{\rho}^{\,2}}}
 \sim t_{T-2}\!\left(0,\frac{1}{T-2}\right),
 \qquad
 \atanh(\hat{\rho})
 \overset{a}{\sim}N\!\left(0,\frac{1}{T-3}\right).
\end{equation*}
The first transformation is the bivariate BCP, and the second is Fisher's
$z$-transformation. The parametrization of
\citet{ArchakovHansen:Correlation} recovers this transformation in the
bivariate case, hence the name generalized Fisher transformation.
Like GFT, BCP extends a familiar scalar transformation from classical
correlation inference to the multivariate setting, providing an
unrestricted parametrization of correlation matrices.

The term \emph{balanced} refers to the balancing condition underlying BCP's
one-to-one mapping from the set $\CorP{n}$ of $n\times n$ positive-definite
correlation matrices onto the space $\Hol{n}$ of symmetric hollow matrices.
The construction
uses the fact that positive diagonal rescaling leaves correlations unchanged.
With $D=\Diag(C^{-1})$, the matrix $D^{1/4}$ is the unique positive diagonal
scaling for which $A=D^{1/4}CD^{1/4}$ and $A^{-1}$ have the same diagonal. This condition makes $B(C)=\frac12(A-A^{-1})$ hollow, while $B+\sqrt{I_n+B^2}$ is
positive definite for every symmetric $B$ and yields a correlation matrix
after normalization. The resulting coordinates also have a statistical
interpretation. Write $K=C^{-1}$, let
$\pi_{ij}=-K_{ij}/\sqrt{K_{ii}K_{jj}}$ be the full partial correlation,
and let $\tau_i=K_{ii}^{-1/2}$ be the residual standard deviation from
projecting standardized variable $i$ on all the other variables.
Each coordinate combines marginal and partial correlations, scaled by
the residual risk of the two variables:
\begin{equation*}
 [B(C)]_{ij}
 =\frac{C_{ij}+\pi_{ij}}{2\sqrt{\tau_i\tau_j}},
 \qquad i\ne j.
\end{equation*}

The main contribution of the paper is to introduce BCP, establish that it
is a global real-analytic bijection, and characterize its construction
through a rigidity theorem. Choosing the diagonal scale before applying
the spectral transformation makes the complete image hollow, so inversion
requires no recovery of an omitted diagonal. For $n\ge3$, direct hollowness
and unit local normalization uniquely determine BCP among primary spectral
transformations of balanced representatives
(\Cref{thm:rigidity-balanced-spectral}). Block symmetry becomes linear
restrictions on the transformed entries and permits computation at the
block level. Exact preservation also yields a new
unrestricted parametrization of block correlation matrices, providing an
alternative to the GFT-based parametrization discussed by
\citet{TongHansen:2023}. We also derive simplified BCP expressions for arrowhead,
Gaussian-tree, moving-average, autoregressive, and factor-model correlation
matrices (\Cref{sec:structured-correlations}), together with explicit
derivatives and a closed-form Jacobian determinant for likelihood
optimization, inference, and density transformation.

The literature on parametrization and sparsity emphasizes that zero
patterns and their interpretation depend on the scale used to represent
dependence \citep{Battey:2017,Battey:2019,RybakBattey:2021,RybakBatteyBharath:2026}.
Related algebraic properties of logarithmic sparsity are studied by
\citet{Pavlov:2024}. \citet{SaxenaBhatPinjari:2023} and
\citet{HuangYeWang:2025} develop Cholesky-based methods for imposing
specified correlation zeros. However, both depend on variable ordering
and require parameter adjustments tailored to the zero pattern.
We show that, under suitable factor conditions, BCP preserves
marginal-correlation sparsity as the dimension tends to infinity: a
coordinate tends to zero if and only if its corresponding correlation
does. Each coordinate is a positively rescaled correlation up to a
uniform error of order $\mu_n^{-1}$, where $\mu_n$ is the strength of
the weakest factor direction relative to residual noise. A common magnitude
threshold can therefore recover the entire marginal-correlation support
under suitable signal and estimation-error conditions. Under the same factor conditions, all GFT and
InvC coordinates shrink uniformly at rates
$O(\log(1+\mu_n)/\mu_n)$ and $O(\mu_n^{-1})$, respectively.
GFT, InvC, and LS do not
share BCP's general support-recovery guarantee: a zero correlation can
produce a larger coordinate magnitude
than a fixed nonzero correlation, even after common rescaling.
For a bounded number of factors, zeroing BCP coordinates at population correlation zeros also yields a
positive-definite reconstruction whose maximum entrywise error tends to
zero. These results provide a basis for sparse high-dimensional
correlation modeling in unrestricted coordinates,
although exact entrywise zeros need not be preserved at finite dimension.

Our empirical illustration uses a high-dimensional balanced panel of
idiosyncratic daily stock returns for 1,765 equities over the 15-year period
2011--2025. We remove the Fama--French five factors and momentum from each
stock's daily return and fit a univariate GJR-GARCH model to the residuals.
We compare the correlations of the resulting standardized shocks
with their BCP, GFT, LS, and InvC coordinates. BCP retains the visible approximate sparsity
pattern, sector structure, and entry magnitudes of shock correlations,
whereas the alternative coordinates track marginal correlations less directly.
The smoothed conditional-coordinate plots support this comparison:
the BCP median tracks shock correlations more closely, while GFT and LS
medians remain closer to zero and the InvC median slopes downward.
We also impose a common zero pattern selected
by pointwise 5\% correlation significance tests. For this sample and these restrictions, BCP yields a
positive-definite estimate that more closely tracks both the original and
directly thresholded correlations and has a lower descriptive BIC than GFT,
LS, or InvC. These findings illustrate how BCP accommodates sparse coordinate
restrictions while yielding a positive-definite reconstruction that closely
approximates the original correlation matrix.

The paper is organized as follows. \Cref{sec:derivation} introduces the
notation, then constructs and characterizes BCP.
\Cref{sec:properties} establishes its structural
properties. \Cref{sec:statistics-geometry} presents differentiation,
inference, and the local relation to GFT.
\Cref{sec:block-symmetry} develops reduced block formulas, and
\Cref{sec:structured-correlations} gives explicit BCP expressions for
selected structured families. Building on the factor-model formulas,
\Cref{sec:high-dimensional-sparsity} establishes asymptotic sparsity
preservation and compares GFT, InvC, and LS.
\Cref{sec:stock-residual-sparsity} presents the standardized-shock illustration. The separate supplement contains
proofs, extensions, and computational details.
Its sections, equations, and tables carry an S prefix.

\begin{table}[htp]
 \centering
 \caption{Comparison of properties of correlation-matrix parametrizations.}
 \label{tab:method-comparison}
 \begin{threeparttable}
 \footnotesize
 \setlength{\tabcolsep}{2.2pt}
 \renewcommand{\arraystretch}{1.18}
 \begin{tabularx}{\textwidth}{@{}l *{6}{>{\centering\arraybackslash}X}@{}}
  \toprule
  Method
  & \shortstack{Closed form\\both ways}
  & \shortstack{Permutation-\\equivariant}
  & \shortstack{Sign-\\equivariant}
  & \shortstack{Inverse-\\consistent}
  & \shortstack{Preserves\\Block}
  & \shortstack{Asymptotic\\Sparsity} \\
  \midrule
  Partial correlation\tnote{$\dagger$} & Yes & No & Yes & No & No & No \\
  GFT & No & Yes & Yes & No & Yes & No \\
  InvC & No & Yes & Yes & No & Yes & No \\
  LS & No & Yes & No & Yes & Yes & No \\
  BCP & \textbf{Yes} & \textbf{Yes} & \textbf{Yes} & \textbf{Yes} & \textbf{Yes} & \textbf{Yes} \\
  \bottomrule
 \end{tabularx}
 \begin{tablenotes}[flushleft]\footnotesize
  \item \textit{Notes:} \tnote{$\dagger$}Fixed-order partial-correlation, Cholesky,
  and hyperspherical parametrizations in sign-equivariant unrestricted
  coordinates. Inverse consistency: $C$ and $\cinv(C)$ have opposite
  coordinates. Preserves block: block structure is preserved in both directions.
  Asymptotic sparsity: under the conditions of \Cref{thm:sparse-bft},
  a coordinate vanishes if and only if the corresponding correlation does
  as $n\rightarrow\infty$; exact zeros
  need not be preserved at finite $n$. ``No'' means not in general.
 \end{tablenotes}
 \end{threeparttable}
\end{table}

\section{Construction and characterization of BCP}\label{sec:derivation}

We construct BCP in two steps. First, we select a unique positive diagonal
rescaling of each correlation matrix whose diagonal matches that of its
inverse. We then characterize the spectral transformation that maps these
balanced matrices to hollow matrices. The resulting transformation and its
inverse are explicit.

\subsection{Notation and preliminaries}\label{sec:background}

Let $I_n$ be the $n\times n$ identity matrix and $\Sym{n}$ the vector space
of symmetric $n\times n$ matrices, with symmetric positive-definite (SPD)
cone $\Symp{n}$. Write
$\CorP{n}=\{C\in\Symp{n}:\Diag(C)=I_n\}$ for the set of
positive-definite correlation matrices and
$\Hol{n}=\{X\in\Sym{n}:\Diag(X)=0\}$ for the symmetric hollow matrices. Both spaces have dimension $d=n(n-1)/2$. For a matrix $X$ and vector $x$, $\diag(X)$ extracts the diagonal vector,
$\diag(x)$ forms a diagonal matrix, and
$\Diag(X)=\diag\{\diag(X)\}$. 

For symmetric $X$, the projection onto $\Hol{n}$ is $\Off(X)=X-\Diag(X)$.
We use $\tr(X)$, $\norm{X}_2$, and $\normF{X}$ for the trace, spectral norm,
and Frobenius norm, respectively, and $\one_n$ for the $n$-vector of ones,
omitting its subscript when dimension is clear.

For $\Sigma\in\Symp{n}$, correlation normalization
$\Cor(\Sigma)=\Diag(\Sigma)^{-1/2}\Sigma\Diag(\Sigma)^{-1/2}$ maps a
covariance matrix to the corresponding correlation matrix.
The \emph{normalized correlation inverse} $\cinv(C)=\Cor(C^{-1})$ gives
the full partial correlations
$\pi_{ij}=-[\cinv(C)]_{ij}=-K_{ij}/\sqrt{K_{ii}K_{jj}}$, $i\ne j$,
where $K=C^{-1}$. A transformation $\Phi$ is \emph{inverse-consistent} if
$\Phi\{\cinv(C)\}=-\Phi(C)$.
The signed permutation group $\sgnperm$ consists of matrices $Q=EP$,
where $P$ is a permutation matrix and
$E=\diag(\varepsilon_1,\ldots,\varepsilon_n)$, $\varepsilon_i\in\{-1,1\}$.

\subsection{The balanced representative}

For each $C\in\CorP{n}$, we seek a positive-definite representative $A_C$
satisfying $\Cor(A_C)=C$ and a scalar spectral transformation $g$ such that
\begin{equation*}
 B=g(A_C)\in\Hol{n},
 \qquad
 C=\Cor\!\left(g^{-1}(B)\right)\in\CorP{n}.
\end{equation*}
Every such representative has the form $A=\Delta C\Delta$ for a positive
diagonal matrix $\Delta$. We select the scale by requiring the matrix and
its inverse to have the same diagonal, a condition preserved by inversion.

\begin{definition}[Balanced matrices]\label{def:balanced}
A symmetric positive-definite matrix $A\in\Symp{n}$ is \emph{balanced} if
its diagonal equals that of its inverse,
\begin{equation}\label{eq:balanced-condition}
 \Diag(A)=\Diag(A^{-1}).
\end{equation}
The set of balanced $n\times n$ matrices is
$\Bal{n}=\{A\in\Symp{n}:\Diag(A)=\Diag(A^{-1})\}$.
\end{definition}

\begin{theorem}[Unique balanced representative]\label{thm:balanced-section}
For every $C\in\CorP{n}$,
there is a unique positive diagonal matrix $\Delta$ such that
$A_C=\Delta C\Delta\in\Bal{n}$. Writing
$D=\Diag(C^{-1})$, the scale and representative are
\begin{equation}\label{eq:AC}
 \Delta=D^{1/4},
 \qquad
 A_C=D^{1/4}CD^{1/4},
\end{equation}
and
\begin{equation}\label{eq:balanced-diagonal}
 \Diag(A_C)=\Diag(A_C^{-1})=D^{1/2}.
\end{equation}
Moreover, for every signed permutation $Q\in\sgnperm$,
\begin{equation}\label{eq:balanced-inversion}
 A_{QCQ'}=QA_CQ',
 \qquad
 A_{\cinv(C)}=A_C^{-1}.
\end{equation}
\end{theorem}

The identity $A_{\cinv(C)}=A_C^{-1}$ shows that normalized correlation
inversion corresponds to inversion of the balanced representative.
The componentwise scaling is the
correlation-matrix specialization of the construction in Lemma~2.4 of
\citet{Fiedler:1964}, which rescales a positive-definite matrix so that its
diagonal equals that of its inverse. The proof of
\Cref{thm:balanced-section}, including the scaling formula and its behavior
under inversion, is given in \Cref{app:balancing-characterization}.

\subsection{Direct hollowness determines the spectral transformation}

With the balanced representative fixed, it remains to choose the spectral
transformation. Let $g:(0,\infty)\to\mathbb R$ be a scalar function.
For an eigendecomposition $A=Q\diag(a_1,\ldots,a_n)Q'$, define its primary
spectral matrix function
by
$g(A) =Q\diag\{g(a_1),\ldots,g(a_n)\}Q'$.
Such a transformation treats all eigendirections symmetrically and, when $g$
is a scalar diffeomorphism, can be inverted by applying $g^{-1}$ to the
eigenvalues.

We require the complete matrix $g(A_C)$ to be hollow for every $C$, so
inversion does not require completing a missing diagonal. We call this
property \emph{direct hollowness}. The normalization $g'(1)=1$ sets the
local scale at independence. For $n\geq3$, these two requirements uniquely
determine $g$.

\begin{theorem}[Rigidity of balanced spectral transformations]
\label{thm:rigidity-balanced-spectral}
Fix $n\geq3$.  Suppose that $g:(0,\infty)\to\mathbb R$ is differentiable at
$1$ and that
\begin{equation}\label{eq:rigidity-hollow}
 \Diag\{g(A)\}=0
 \qquad\text{for every }A\in\Bal{n}.
\end{equation}
Then
\begin{equation*}
 g(a)=\frac{g'(1)}{2}\left(a-a^{-1}\right),
 \qquad a>0.
\end{equation*}
Consequently, under the unit local normalization $g'(1)=1$,
\begin{equation*}
 g(A)=\frac12(A-A^{-1}),
 \qquad
 g(a)=\frac12(a-a^{-1})=\sinh(\log a).
\end{equation*}
\end{theorem}

Differentiability at $1$ is the only regularity assumption. When
$g'(1)\ne0$, analyticity, global bijectivity, and the inversion identity
$g(a^{-1})=-g(a)$ follow from the characterization. The proof in
\Cref{app:balancing-characterization} reduces direct hollowness to a scalar
functional equation. The restriction $n\geq3$ is essential: in dimension
two, direct hollowness only requires $g(a^{-1})=-g(a)$, leaving infinitely
many normalized choices, including $g(a)=\log a$.

The characterization concerns primary spectral transformations of balanced
representatives. Entrywise odd analytic reparametrizations can preserve the
properties in \Cref{tab:method-comparison} without belonging to this class;
\Cref{rem:postcomposition-family} gives an explicit family.

Up to the factor $\frac12$, $g(A)=\frac12(A-A^{-1})$ is the covariance
transformation $\Sigma\mapsto\Sigma-\Sigma^{-1}$ that \citet{Zwiernik:2025}
studies as an algebraic alternative to the matrix logarithm. The unscaled
map $C-C^{-1}$ has diagonal entries $1-(C^{-1})_{ii}\le0$, which are
constrained by the correlations. Applying $g$ to the balanced representative
$A_C$ yields a hollow matrix.

\subsection{The transformation and its inverse}

Balance selects $A_C$, and the rigidity theorem determines
$g(a)=(a-a^{-1})/2$. Its scalar inverse, $g^{-1}(b)=b+\sqrt{1+b^2}$,
gives the inverse matrix transformation directly.

\begin{theorem}[Balanced correlation parametrization]\label{thm:global-diffeo}
For $C\in\CorP{n}$, let $A_C$ be its balanced representative in
\eqref{eq:AC}. The \emph{balanced correlation parametrization} is the global
real-analytic diffeomorphism $C\mapsto B(C)$ from $\CorP{n}$ onto $\Hol{n}$,
with inverse $B\mapsto C(B)$, given by
\begin{equation}\label{eq:BFT-main-definition}
 \begin{aligned}
  B(C)
  &=\tfrac12(A_C-A_C^{-1})=\sinh(\log A_C),\\[2pt]
  C(B)
  &=\Cor(B+\sqrt{I_n+B^2})
   =\Cor(e^{\arsinh B}),\qquad B\in\Hol{n}.
 \end{aligned}
\end{equation}
Here $\sinh$ and $\arsinh$ denote the corresponding matrix functions
defined by spectral calculus. The matrix $B+\sqrt{I_n+B^2}$ is the unique
balanced SPD matrix $A$ satisfying
$\frac12(A-A^{-1})=B$. Thus the balanced-to-hollow map is itself a
real-analytic bijection.
\end{theorem}

Correlation normalization removes the scale used to construct $A_C$,
so both directions are explicit and require no iterative diagonal correction.
For vector coordinates, $\vecl$ stacks the strict lower-triangular entries,
giving $\beta(C)=\vecl\{B(C)\}\in\mathbb R^d$.
The reconstruction identities and proof appear in
\Cref{app:global-diffeo-proof} of the supplement.

\begin{remark}[Computation]
The forward map uses an SPD solve, diagonal extraction, and
matrix products.  The inverse uses one symmetric eigendecomposition and a
diagonal normalization.  Neither map requires an iterative scaling or diagonal-completion step.  For numerical stability, if
$B=Q\diag(\lambda_i)Q'$, evaluate
$B+\sqrt{I_n+B^2}=Q\diag\!\left(e^{\arsinh\lambda_i}\right)Q'$, which avoids
cancellation for large negative eigenvalues.
\end{remark}

\section{Structural properties}\label{sec:properties}

The balanced construction preserves the natural symmetries of correlation
matrices and gives an exact relation between correlation and precision.
We first establish equivariance and inverse consistency, then interpret
individual coordinates and study block structure, separability, and
behavior near the singular boundary. Proofs appear in
\Cref{app:proofs} of the supplement.

\subsection{Signed-permutation equivariance}

\begin{theorem}[Signed-permutation equivariance]\label{thm:signed-equivariance}
For every $Q\in\sgnperm$, every $C\in\CorP{n}$, and every $B\in\Hol{n}$,  we have
\begin{equation}\label{eq:BFT-equivariance}
 B(QCQ')=QB(C)Q',
 \qquad C(QBQ')=QC(B)Q'.
\end{equation}
\end{theorem}

Reordering variables or changing their signs applies the same permutation
or sign changes to the transformed entries. The definition therefore
respects the labeling and sign conventions of the original variables.

\subsection{Exact inverse consistency}

\begin{theorem}[Correlation inversion becomes sign reversal]\label{thm:inverse-consistency}
For every $C\in\CorP{n}$ and every $B\in\Hol{n}$, we have
\begin{equation}\label{eq:inverse-consistency}
 B(\cinv(C))=-B(C),\qquad C(-B)=\cinv(C(B)).
\end{equation}
\end{theorem}

This distinguishes BCP from GFT, which
is not inverse-consistent when $n\ge3$ \citep{Thanwerdas:2024}. Inverse consistency is not required for a valid parametrization.  It matters
when transformed matrices are averaged, interpolated, or regularized.  The matrices $C$ and $\cinv(C)$ describe the same dependence
system through correlation and normalized precision.  BCP places them at
opposite points, $B$ and $-B$, in the same Euclidean space. Penalties
symmetric about zero therefore treat the two descriptions alike, while
averaging and interpolation commute with sign reversal.  In particular, along the
straight line in coordinates from $B(C)$ to $B(\cinv(C))=-B(C)$, the midpoint is
$B=0$, which corresponds to $C=I_n$.

\subsection{Marginal and partial correlations}

Each BCP entry has a direct interpretation in terms of marginal and full
partial correlations.

\begin{proposition}[Marginal--partial representation]\label{prop:partial-correlation}
Let $C\in\CorP{n}$ and $K=C^{-1}$. Write
$\pi_{ij}=-K_{ij}/\sqrt{K_{ii}K_{jj}}$ for the full partial correlation
and $\tau_i=K_{ii}^{-1/2}$ for the residual standard deviation of variable
$i$ after linear projection on the remaining variables. Then
\begin{equation}\label{eq:B-R2}
 [B(C)]_{ij}
 =\frac{C_{ij}+\pi_{ij}}{2\sqrt{\tau_i\tau_j}},\quad i\ne j.
\end{equation}
\end{proposition}

The sign of $[B(C)]_{ij}$ is the sign of $C_{ij}+\pi_{ij}$.  Holding the residual standard deviations fixed, its magnitude increases
with the absolute value of this sum; holding the sum fixed, it increases
as the residual standard deviations decrease. Thus each BCP coordinate is a
residual-risk-adjusted average of marginal and partial correlations, rather
than an entrywise transformation of the marginal correlation alone.

\subsection{Preservation of blocks and symmetries}

\begin{proposition}[Direct sums and fixed-point classes]\label{prop:blocks}
Let $\oplus$ denote the block direct sum.  BCP and its
inverse preserve direct sums:
$B(C_1\oplus C_2) =B(C_1)\oplus B(C_2)$,
and
$C(B_1\oplus B_2) =C(B_1)\oplus C(B_2)$.
More generally, for any subgroup $\mathfrak G\subseteq\sgnperm$,
\[
 QCQ'=C\ \forall Q\in\mathfrak G
 \quad\Leftrightarrow\quad
 QB(C)Q'=B(C)\ \forall Q\in\mathfrak G.
\]
\end{proposition}

Thus equicorrelation and other exchangeable or block-symmetric classes map to
linear fixed-point subspaces of $\Hol{n}$. Permutation-generated symmetry
restrictions also arise in graphical Gaussian models
\citep{HojsgaardLauritzen:2008}. \Cref{sec:block-symmetry} gives the
corresponding block formulas and a reduced $G\times G$ implementation based on
the canonical representation of \citet{ArchakovHansen:CanonicalBlockMatrix}.

\subsection{Separable correlation matrices}

Kronecker products are widely used for separable panel, spatial, and
spatio-temporal dependence.  BCP preserves their low Kronecker complexity.

\begin{proposition}[Kronecker-product identity]\label{prop:kronecker}
Let $C_k\in\CorP{n_k}$, $S_k=B(C_k)$, and
$H_k=\sqrt{I_{n_k}+S_k^2}$ for $k=1,2$.
The balanced representative and BCP satisfy
\begin{equation}\label{eq:BFT-kronecker}
 A_{C_1\otimes C_2}=A_{C_1}\otimes A_{C_2},
 \qquad B(C_1\otimes C_2)=S_1\otimes H_2+H_1\otimes S_2.
\end{equation}
Thus a separable correlation matrix maps to a sum of two Kronecker products.
\end{proposition}

\subsection{Boundary behavior}

The next result quantifies how approaching the singular boundary affects
the size of the BCP coordinates.

\begin{proposition}[Properness and a two-sided bound]\label{prop:properness}
The map $C\mapsto B(C)$ from $\CorP{n}$ to $\Hol{n}$ is proper.  Quantitatively, write $\ell=\lambda_{\min}(C)$ and $B=B(C)$.  Then
\begin{equation}\label{eq:properness-bound}
 \frac12\left(\ell^{-1/2}-\ell^{1/2}\right)
 \leq\norm{B}_2
 \leq\frac12\left(\frac n\ell-\frac\ell n\right),
\end{equation}
and the left inequality is equivalent to
\begin{equation}\label{eq:lambda-min-bound}
 \lambda_{\min}(C)\geq e^{-2\arsinh\norm{B}_2}.
\end{equation}
In particular,
\begin{equation}\label{eq:boundary-divergence}
 \lambda_{\min}(C_k)\rightarrow0
 \quad\Rightarrow\quad
 \normF{B(C_k)}\rightarrow\infty.
\end{equation}
Conversely, a uniform bound on the BCP norm keeps the reconstructed
correlation matrices in a compact subset of the open elliptope and hence
uniformly away from singularity.
\end{proposition}

The left inequality in \eqref{eq:properness-bound} is sharp up to a factor
$\sqrt2$, attained in the limit $n=2$, $\rho\to1$.  The bound relates the size
of the transformed matrix directly to the conditioning of the correlation
matrix it generates.

\section{Differentiation and inference}
\label{sec:statistics-geometry}

Both directions of the transformation can be differentiated without solving
an implicit diagonal-completion problem.  Likelihood
optimization, filtering, and score-driven updates usually require the action
of a derivative on a direction. The formulas below evaluate that action
directly, avoiding the full $d\times d$ Jacobian, where $d=n(n-1)/2$.  The corresponding adjoint and likelihood-score formulas are given in
\Cref{app:adjoint-geometry} of the supplement. We first give the forward
and inverse derivatives and a closed-form Jacobian determinant, then
derive delta-method inference and compare
BCP with GFT near independence.

\subsection{Closed-form differential operations}

\begin{proposition}[Forward derivative]\label{prop:differential-forward}
Let $C\in\CorP{n}$, $X\in\Hol{n}$, and define
$K=C^{-1}$, $D=\Diag(K)$,
$\Delta=D^{1/4}$ and $A=\Delta C\Delta$.  Using $\oslash$ for
componentwise division, set
\begin{align}
 U_X&=-\frac14\diag\{\diag(KXK)\oslash\diag(K)\},\notag
\\
 Z_X&=\Delta(X+U_XC+CU_X)\Delta.
 \label{eq:ZX-forward}
\end{align}
Then the Jacobian--vector product is
\begin{equation}\label{eq:dPsiBFT-compact}
 \dd_C B[X]
 =\frac12\left(Z_X+A^{-1}Z_XA^{-1}\right).
\end{equation}
In particular,
\begin{equation}\label{eq:derivative-identity}
 \dd_{I_n}B[X]=X.
\end{equation}
\end{proposition}

The calculation reuses the solve needed for the transformation and otherwise
requires only diagonal extraction, scaling, and matrix products.

\begin{proposition}[Inverse derivative]\label{prop:differential-inverse}
Let $B=Q\diag(\lambda_1,\ldots,\lambda_n)Q'$ and define
\begin{equation}\label{eq:ai-omega}
 a_i=e^{\arsinh(\lambda_i)},\qquad
 \Omega_{ij}=\frac{2a_ia_j}{1+a_ia_j}.
\end{equation}
Write $A=Q\diag(a_i)Q'$, $e=\diag(A)$,
$R=\diag(e^{-1/2})$ and $C=RAR$.  For $Y\in\Hol{n}$, let
\begin{equation}\label{eq:ZY-inverse}
 Z_Y=Q\{\Omega\Had(Q'YQ)\}Q',
 \qquad v_Y=\diag(Z_Y)\oslash e,\qquad V_Y=\diag(v_Y).
\end{equation}
Here $\circ$ denotes entrywise multiplication. Then
\begin{equation}\label{eq:dPsiBFTinv-explicit}
 \dd_B C[Y]
 =R\left\{Z_Y-\frac12(V_YA+AV_Y)\right\}R.
\end{equation}
Moreover, $0<\Omega_{ij}<2$ for all $i,j$.
\end{proposition}

The coefficients in \eqref{eq:ai-omega} are bounded and require no special
branch for repeated eigenvalues.  The same eigendecomposition evaluates both
the inverse map and its derivative.

\begin{proposition}[Jacobian determinant]\label{prop:jacobian-determinant}
For $B\in\Hol{n}$, retain $e_i$ and $\Omega_{ij}$ from
\Cref{prop:differential-inverse}. With $\vecl$ denoting strict
lower-triangle vectorization, the inverse Jacobian determinant is
\begin{equation}\label{eq:jacobian-determinant}
 j_n(B):=
 \left|\det\frac{\partial\vecl\{C(B)\}}
 {\partial\vecl(B)'}\right|
 =\prod_{i=1}^n e_i^{-(n-1)/2}
  \prod_{1\leq i\leq j\leq n}\Omega_{ij}.
\end{equation}
For the forward Jacobian
$J_C=\partial\vecl\{B(C)\}/\partial\vecl(C)'$,
$|\det J_C|=j_n\{B(C)\}^{-1}$.
\end{proposition}

Once the inverse map has been evaluated, the determinant requires only
$O(n^2)$ additional operations, with no full $d\times d$ Jacobian.
For numerical evaluation, its logarithm is computed by summing the logs
of the positive factors. The formula gives $j_n(0)=1$ and, for $n=2$
with $B_{12}=b$, $j_2(B)=(1+b^2)^{-3/2}$.
It is useful for change-of-variables densities; reparametrizing a likelihood
as $\ell\{C(B)\}$ does not require a Jacobian-determinant term.

\subsection{Delta-method inference}

Let $J_C$ be the forward Jacobian defined in
\Cref{prop:jacobian-determinant}. Applying the standard delta method to the
analytic map $C\mapsto B(C)$ transfers asymptotic results for a sample correlation
matrix directly to the transformed parameters, as follows.

\begin{corollary}[Delta method]\label{thm:delta}
For fixed $n$, if $\hat C_T\in\CorP{n}$ satisfies
$\sqrt T\{\vecl(\hat C_T)-\vecl(C)\} \xrightarrow{d}N(0,V_C)$,
then
\begin{equation}\label{eq:B-CLT}
 \sqrt T\left[
 \vecl\{B(\hat C_T)\}-\vecl\{B(C)\}
 \right]
 \xrightarrow{d}N(0,J_CV_CJ_C').
\end{equation}
At independence, $J_{I_n}=I_d$.
\end{corollary}

\subsection{Relation to the generalized Fisher transformation}\label{sec:AH}

The GFT of \citet{ArchakovHansen:Correlation} uses
$\Off(\log C)$ as unrestricted coordinates for the same correlation
matrices. Its inverse
recovers an omitted diagonal, whereas BCP uses the explicit balanced
representative. Despite this difference, the maps agree through second
order around independence.

\begin{theorem}\label{thm:local-comparison}
Fix $n$ and let $\norm{\cdot}$ be a submultiplicative matrix norm. For $C=I_n+E\in\CorP{n}$,
with $E\in\Hol{n}$ and $\norm E$ sufficiently small,
\begin{equation}\label{eq:BFT-OL-difference}
 B(C)-\Off(\log C)
 =\frac16\Off(E^3)+\cO(\norm E^4).
\end{equation}
Both transformations equal
$E-\frac12\Off(E^2)+\cO(\norm E^3)$.
The remainder constants may depend on the fixed dimension $n$.
\end{theorem}

The proof and the separate expansions appear in \Cref{app:expansions}.
The common linear term gives the same Jacobian, $I_d$, at independence
and hence the same first-order delta-method covariance. More precisely,
if $\hat C_T-I_n=\cO_p(T^{-1/2})$, then
$B(\hat C_T)-\Off(\log\hat C_T)=\cO_p(T^{-3/2})$
on the event that $\hat C_T$ is positive definite, whose probability
tends to one. Near independence, the Jacobians differ at quadratic order
in $\norm{C-I_n}$.

This local agreement does not extend to equality of the transformations:
BCP retains explicit inversion and inverse consistency globally.
Supplementary \Cref{app:AH-relations} gives the exact transition map and
its simplification when the diagonal scale commutes with the correlation
matrix. The local comparison above is controlled by a matrix norm at
fixed dimension; small individual correlations in a large matrix need not
make that norm small. The high-dimensional sparsity results in
\Cref{sec:high-dimensional-sparsity} therefore require separate structural
conditions.

\section{Block correlation matrices}\label{sec:block-symmetry}

Block structures arise naturally with sectors, clusters, and exchangeable
groups.  BCP preserves a fixed block partition in both directions, allowing the dense matrix operations to use a matrix whose dimension is
the number of blocks. Exact preservation holds for all block sizes, including
singletons. For clarity, the reduced forward and inverse formulas below
are stated for blocks with at least two members; Supplementary
Section~\ref{app:block-singletons} gives the corresponding singleton formulas.

\subsection{A three-block illustration}

We begin with a simple example. Consider seven variables divided into three
 groups of sizes two, two, and three, with common correlations within each
 group and between each pair of groups. The correlation matrix and its BCP
 coordinates are\par
\begingroup
\setstretch{1}
\small
\setlength{\arraycolsep}{2.5pt}
\renewcommand{\arraystretch}{1.15}
\begin{equation*}
 \underbrace{\left[\begin{array}{ccccccc}
\cellcolor{black!07}1.0 & \cellcolor{black!07}0.8 & 0.4 & 0.4 & \cellcolor{black!07}0.2 & \cellcolor{black!07}0.2 & \cellcolor{black!07}0.2\\
\cellcolor{black!07}0.8 & \cellcolor{black!07}1.0 & 0.4 & 0.4 & \cellcolor{black!07}0.2 & \cellcolor{black!07}0.2 & \cellcolor{black!07}0.2\\
0.4 & 0.4 & \cellcolor{black!07}1.0 & \cellcolor{black!07}0.6 & 0.1 & 0.1 & 0.1\\
0.4 & 0.4 & \cellcolor{black!07}0.6 & \cellcolor{black!07}1.0 & 0.1 & 0.1 & 0.1\\
\cellcolor{black!07}0.2 & \cellcolor{black!07}0.2 & 0.1 & 0.1 & \cellcolor{black!07}1.0 & \cellcolor{black!07}0.3 & \cellcolor{black!07}0.3\\
\cellcolor{black!07}0.2 & \cellcolor{black!07}0.2 & 0.1 & 0.1 & \cellcolor{black!07}0.3 & \cellcolor{black!07}1.0 & \cellcolor{black!07}0.3\\
\cellcolor{black!07}0.2 & \cellcolor{black!07}0.2 & 0.1 & 0.1 & \cellcolor{black!07}0.3 & \cellcolor{black!07}0.3 & \cellcolor{black!07}1.0
 \end{array}\right]}_{=\ C}
 \quad \Longleftrightarrow \quad
 \underbrace{\left[\begin{array}{ccccccc}
\cellcolor{black!07}0 & \cellcolor{black!07}1.303 & 0.355 & 0.355 & \cellcolor{black!07}0.163 & \cellcolor{black!07}0.163 & \cellcolor{black!07}0.163\\
\cellcolor{black!07}1.303 & \cellcolor{black!07}0 & 0.355 & 0.355 & \cellcolor{black!07}0.163 & \cellcolor{black!07}0.163 & \cellcolor{black!07}0.163\\
0.355 & 0.355 & \cellcolor{black!07}0 & \cellcolor{black!07}0.715 & 0.062 & 0.062 & 0.062\\
0.355 & 0.355 & \cellcolor{black!07}0.715 & \cellcolor{black!07}0 & 0.062 & 0.062 & 0.062\\
\cellcolor{black!07}0.163 & \cellcolor{black!07}0.163 & 0.062 & 0.062 & \cellcolor{black!07}0 & \cellcolor{black!07}0.277 & \cellcolor{black!07}0.277\\
\cellcolor{black!07}0.163 & \cellcolor{black!07}0.163 & 0.062 & 0.062 & \cellcolor{black!07}0.277 & \cellcolor{black!07}0 & \cellcolor{black!07}0.277\\
\cellcolor{black!07}0.163 & \cellcolor{black!07}0.163 & 0.062 & 0.062 & \cellcolor{black!07}0.277 & \cellcolor{black!07}0.277 & \cellcolor{black!07}0
 \end{array}\right]}_{=\ B(C)}.
\end{equation*}
\endgroup
The gray shading identifies the same blocks in both matrices.
The transformation preserves the three-block pattern exactly: entries
that are equal under the original partition remain equal after
transformation. The remainder of the section establishes this
preservation generally and develops reduced formulas that operate directly on
the distinct block parameters.

\subsection{Correlation and transformed block classes}

Fix a partition $\mathcal P=(n_1,\ldots,n_G)$ of $n$, with $n_k\ge1$ for
every $k$.  A block correlation matrix has blocks
\begin{equation*}
 C_{[k,k]}=(1-\rho_{kk})I_{n_k}
 +\rho_{kk}\one_{n_k}\one_{n_k}',
 \qquad
 C_{[k,l]}=\rho_{kl}\one_{n_k}\one_{n_l}'\quad(k\ne l),
\end{equation*}
subject to $C\succ0$.  Denote this class by $\CorBlock{n}$.  The corresponding
linear space $\HolBlock{n}$ consists of hollow matrices with
\begin{equation}\label{eq:block-hollow-definition}
 B_{[k,k]}=\eta_{kk}
 (\one_{n_k}\one_{n_k}'-I_{n_k}),
 \qquad
 B_{[k,l]}=\eta_{kl}\one_{n_k}\one_{n_l}'\quad(k\ne l).
\end{equation}
For a singleton block, these expressions reduce to $C_{[k,k]}=[1]$ and
$B_{[k,k]}=[0]$, with no within-block parameter. Both classes have
$G(G-1)/2+|\{k:n_k>1\}|$ free parameters: one for each unordered pair
of blocks and one within each nonsingleton block.

\Needspace{5\baselineskip}
\subsection{Exact preservation}\label{sec:block-preservation}

\begin{theorem}[Exact preservation of block symmetry]
\label{thm:block-restriction}
For every fixed partition with $n_k\ge1$, the transformation and its inverse
restrict to mutually inverse real-analytic diffeomorphisms
\begin{equation*}
 \CorBlock{n}
 \overset{B(\cdot)}{\underset{C(\cdot)}{\rightleftarrows}}
 \HolBlock{n}.
\end{equation*}
Thus the partition and its constant-entry pattern are preserved in both
directions.
\end{theorem}

This preservation follows directly from equivariance.  The two block classes
are the fixed-point sets of independent permutations within each block, and a
bijective map that commutes with these permutations must map one fixed-point
set onto the other.  The same argument applies to fixed-point classes of other
subgroups of signed permutations.

For fixed block sizes $n_k\ge2$, \Cref{thm:block-restriction} provides a new
parametrization of block correlation matrices by an unrestricted symmetric
$G\times G$ matrix $(\eta_{kl})$. Its diagonal entries specify within-block
coordinates, while its off-diagonal entries specify between-block coordinates.
Every choice yields a unique positive-definite correlation matrix with the
prescribed block structure. This provides an alternative to the GFT-based
block parametrization discussed by \citet{TongHansen:2023}, with both the
forward and inverse maps available in closed form. For singleton blocks,
the corresponding within-block coordinates are omitted.

\subsection{Reduced forward transformation}

For the reduced formulas in this subsection, assume $n_k\ge2$ for every $k$.
Define the $G\times G$ core $R=(r_{kl})$ and the within-block contrast
eigenvalues $\lambda_k$ by
$r_{kk}=1+(n_k-1)\rho_{kk}$, $r_{kl}=\rho_{kl}\sqrt{n_kn_l}$ for $k\ne l$, and $\lambda_k=1-\rho_{kk}$.
Positive definiteness of $C$ is equivalent to $R\succ0$ and
$\lambda_k>0$ for every $k$.

\begin{theorem}[Reduced forward transformation]
\label{thm:block-forward}
For $C\in\CorBlock{n}$, define $\delta_k>0$ and $\Delta$ by
\begin{equation}\label{eq:block-dk}
 \delta_k^4=\frac{(R^{-1})_{kk}}{n_k}
 +\frac{n_k-1}{n_k\lambda_k},\qquad
 \Delta=\diag(\delta_1,\ldots,\delta_G).
\end{equation}
Then $\Diag(C^{-1})$ is constant within block $k$, with value $\delta_k^4$.  Set
\begin{equation}\label{eq:block-forward-components}
 M =\frac12\left(\Delta R\Delta-
                 \Delta^{-1}R^{-1}\Delta^{-1}\right), \qquad 
 \zeta_k =\frac12\left(\delta_k^2\lambda_k-
                 \delta_k^{-2}\lambda_k^{-1}\right).
\end{equation}
The entries of $B(C)$ in \eqref{eq:block-hollow-definition} are
\begin{equation}\label{eq:block-eta-forward}
 \eta_{kl}=\frac{M_{kl}}{\sqrt{n_kn_l}}\quad(k\ne l),
 \qquad
 \eta_{kk}=-\zeta_k=\frac{M_{kk}}{n_k-1}.
\end{equation}
\end{theorem}

Thus the forward transformation requires one $G\times G$ inverse and scalar
operations.  The block means interact through $R$, while each block contributes
one scalar contrast.

\subsection{Reduced inverse transformation}

Again assume $n_k\ge2$ for every $k$.
Given $B\in\HolBlock{n}$, form the reduced core
\begin{equation}\label{eq:block-H-zeta-input}
 H_{kk}=(n_k-1)\eta_{kk},\qquad
 H_{kl}=\eta_{kl}\sqrt{n_kn_l}\ (k\ne l),\qquad
 \zeta_k=-\eta_{kk}.
\end{equation}

\begin{theorem}[Reduced inverse transformation]
\label{thm:block-inverse}
Let $f(x)=x+\sqrt{1+x^2}=e^{\arsinh x}$ and define
$F=f(H)=H+\sqrt{I_G+H^2}$ and $a_k=f(\zeta_k)$.
Set
\begin{equation}\label{eq:block-ek}
 e_k=\frac{F_{kk}+(n_k-1)a_k}{n_k}.
\end{equation}
Then $C=C(B)$ has block correlations
\begin{equation}\label{eq:block-inverse-correlations}
 \rho_{kk}=\frac{F_{kk}-a_k}{n_ke_k},\qquad
 \rho_{kl}=\frac{F_{kl}}{\sqrt{n_kn_le_ke_l}},\quad k\ne l.
 \end{equation}
\end{theorem}

\begin{corollary}[Reduced computational dimension]
\label{cor:block-complexity}
If only the distinct block parameters are required, the forward map uses one
$G\times G$ inverse and the inverse map uses one $G\times G$ principal square
root, followed by scalar operations.  Their dense matrix-function cost is
$O(G^3)$ rather than $O(n^3)$.  Constructing the full matrix still requires
$O(n^2)$ output work.
\end{corollary}

\subsection{Equicorrelation matrices}
\label{sec:equicorrelation-calibration}

The one-block case, $G=1$, corresponds to an equicorrelation matrix.
It yields explicit scalar formulas for BCP and its inverse and provides
a simple comparison with GFT.

\begin{proposition}[BCP for equicorrelation matrices]\label{prop:equicorrelation}
Let $n\ge2$ and $C_n(\rho)=(1-\rho)I_n+\rho\one\one'$, with
$-1/(n-1)<\rho<1$. Then
$B(C_n(\rho)) =b_n(\rho)(\one\one'-I_n)$,
where
\begin{equation}\label{eq:BFT-equicorr}
 b_n(\rho)
 =\frac{\rho\{2+(n-2)\rho\}}
 {2\sqrt{\{1+(n-2)\rho\}(1-\rho)\{1+(n-1)\rho\}}}.
\end{equation}
The inverse is also explicit. Given $b\in\R$, set
$a_1=(n-1)b+\sqrt{1+(n-1)^2b^2}$ and $a_2=-b+\sqrt{1+b^2}$.
Then
\begin{equation}\label{eq:rho-from-b}
 \rho=\frac{a_1-a_2}{a_1+(n-1)a_2}.
\end{equation}
\end{proposition}

Let $z(\rho)=\atanh(\rho)$ denote Fisher's $z$-transformation.
The bivariate BCP takes the form
\begin{equation*}
 B\!\left(
 \begin{bmatrix}1&\rho\\\rho&1\end{bmatrix}
 \right)
 =\sinh\!\left(
 \begin{bmatrix}0&z(\rho)\\z(\rho)&0\end{bmatrix}
 \right)
 =\begin{bmatrix}
 0&\tfrac{\rho}{\sqrt{1-\rho^2}}\\[6pt]
 \tfrac{\rho}{\sqrt{1-\rho^2}}&0
 \end{bmatrix}.
\end{equation*}
Thus, the BCP coordinate is the hyperbolic sine of Fisher's
$z$-transformation, whereas the GFT coordinate equals Fisher's $z$ itself.

For general $n$, the common GFT coordinate of an equicorrelation matrix is
$g_n(\rho)=n^{-1}\log\{[1+(n-1)\rho]/(1-\rho)\}$.
The exact relation between the two coordinates is
\begin{equation}\label{eq:AH-BFT-equicorr-link}
  g_n(\rho)
 =\frac{1}{n}\left[\arsinh\bigl((n-1)b_n(\rho)\bigr)
 +\arsinh\bigl(b_n(\rho)\bigr)\right].
\end{equation}
For fixed $n$, these coordinates agree through second order as $\rho\to0$,
in accordance with the general local comparison in \Cref{sec:AH}.
A proof of \Cref{prop:equicorrelation} and identity
\eqref{eq:AH-BFT-equicorr-link} is given in Supplementary
\Cref{app:equicorrelation}. The identity is a scalar instance of the general
commuting relation in Supplementary \Cref{prop:commuting-link}.

\section{BCP for selected structured correlation matrices}
\label{sec:structured-correlations}

BCP is explicit for every positive-definite correlation matrix, but its forward
formula becomes particularly transparent when the precision matrix is known
analytically.  Let $K=C^{-1}$ and
$\delta_i=K_{ii}^{1/4}$.  Directly from the definition, $[B(C)]_{ij}
 =\frac12(
  \delta_i\delta_j C_{ij}
  -\frac{1}{\delta_i\delta_j}K_{ij}
 )$, $i\ne j$. Thus a structured precision matrix immediately gives a simplified BCP
expression, without a matrix eigendecomposition in the forward direction.
We consider arrowhead, Gaussian-tree, moving-average, autoregressive, and
standard factor correlations. The time-series examples distinguish banded
correlations from banded precision matrices.

\subsection{Arrowhead correlation matrices}

For $\rho\in\R^r$ with $\rho'\rho<1$, consider the $(r+1)\times(r+1)$
arrowhead correlation matrix
\begin{equation}\label{eq:star-correlation-matrix}
 C^\star(\rho)
 =\begin{pmatrix}
  1 & \rho'\\
  \rho & I_r
 \end{pmatrix}.
\end{equation}
This is a unit-diagonal instance of the symmetric arrowhead matrices
studied by \citet{OLearyStewart:1990}. It is also the auxiliary correlation
matrix for $(Z,U')'$ used by
\citet{TongHansen:2025FactorGARCH}, who derive a tractable expression
under GFT.

\begin{proposition}[BCP coordinates of an arrowhead correlation matrix]
\label{prop:star-coordinate}
Let
$\omega=\sqrt{1-\rho'\rho}$ and $h_j=\omega^2+\rho_j^2$ for $j=1,\ldots,r$,
and define
$R_\rho=\diag\left(h_1^{1/4},\ldots,h_r^{1/4}\right)$ and $v_\rho=R_\rho^{-1}\rho$.
Then
\begin{equation}\label{eq:BFT-star-block-form}
 B\left(C^\star(\rho)\right)
 =\begin{pmatrix}
  0 & \eta_\rho'\\
  \eta_\rho & \Psi_\rho
 \end{pmatrix},
\end{equation}
where
\begin{equation*}
 \eta_\rho
 =\frac{1}{2\omega}\left(R_\rho+R_\rho^{-1}\right)\rho,
 \qquad
 \Psi_\rho
 =-\frac{1}{2\omega}\Off\left(v_\rho v_\rho'\right).
\end{equation*}
Equivalently, for $j\ne k$,
\begin{equation}
 [\eta_\rho]_j
 =\frac{\rho_j}{2\omega}
   \left(h_j^{1/4}+h_j^{-1/4}\right),
 \qquad  
 [\Psi_\rho]_{jk}=-\frac{\rho_j\rho_k}
 {2\omega(h_jh_k)^{1/4}},
 \qquad [\Psi_\rho]_{jj}=0.
 \label{eq:BFT-star-factor-elements}
\end{equation}
\end{proposition}

The vector $\eta_\rho$ preserves the zero pattern and signs of $\rho$.
The block $\Psi_\rho$ records linear partial associations among the peripheral
variables, although their block in $C^\star(\rho)$ is the identity. This illustrates why a zero marginal correlation need not
produce a zero BCP coordinate at finite dimension.

\subsection{Gaussian-tree correlation matrices}
\label{subsec:tree-correlations}

Let $(V,E)$ be a tree and associate an edge correlation
$\theta_{ij}\in(-1,1)$ with every $\{i,j\}\in E$.  The corresponding
standardized Gaussian tree has
\begin{equation}\label{eq:tree-correlation}
 C_{ij}=\prod_{e\in\operatorname{path}(i,j)}\theta_e,
 \qquad C_{ii}=1.
\end{equation}
Its precision matrix $K=C^{-1}$ is supported on the tree:
\begin{equation}\label{eq:tree-precision}
 K_{ij}=
 \begin{cases}
  -\dfrac{\theta_{ij}}{1-\theta_{ij}^2},
    &\{i,j\}\in E,\\[6pt]
  0,&\{i,j\}\notin E,
 \end{cases}
 \quad i\ne j,\quad \text{and} \quad  K_{ii}
 =1+\sum_{j:\{i,j\}\in E}
 \frac{\theta_{ij}^2}{1-\theta_{ij}^2}.
\end{equation}

\begin{proposition}[BCP coordinates of a Gaussian-tree correlation matrix]
\label{prop:tree-correlation}
Let $\delta_i=K_{ii}^{1/4}$.  For $i\ne j$,
\begin{equation}\label{eq:BFT-tree-correlation}
 [B(C)]_{ij}
 =\begin{cases}
 \dfrac12\left\{
  \delta_i\delta_j\theta_{ij}
  +\dfrac{\theta_{ij}}
   {(1-\theta_{ij}^2)\delta_i\delta_j}
 \right\},&\{i,j\}\in E,\\[10pt]
 \dfrac12\delta_i\delta_j
  \displaystyle\prod_{e\in\operatorname{path}(i,j)}\theta_e,
  &\{i,j\}\notin E.
 \end{cases}
\end{equation}
\end{proposition}

The homogeneous path tree gives the AR(1) correlation matrix
$C_{ij}=\rho^{|i-j|}$, which is a special case of the autoregressive
family considered in \Cref{subsec:ar-correlations}.

\subsection{Moving-average correlation matrices}
\label{subsec:ma-correlations}

Consider the MA($q$) process
$x_t=\sum_{\ell=0}^q\theta_\ell\varepsilon_{t-\ell}$, where
$\theta_0=1$ and $\varepsilon_t$ is zero-mean, unit-variance white noise.
For $n>q\ge1$, the correlation matrix of $(x_1,\ldots,x_n)'$ is banded
Toeplitz, with $C_{ij}=\rho_{|i-j|}$ and
\begin{equation}\label{eq:MAq-correlations}
 s=\sum_{\ell=0}^q\theta_\ell^2,
 \qquad
 \rho_h=
 \begin{cases}
 s^{-1}\displaystyle\sum_{\ell=0}^{q-h}
       \theta_\ell\theta_{\ell+h},&0\le h\le q,\\[4pt]
 0,&h>q.
 \end{cases}
\end{equation}
This is the usual finite-order moving-average covariance structure
\citep[Chapter~3]{Hamilton:1994}. Its inverse can be obtained from a
finite filter recursion and a correction for the $q$ presample innovations.
Define
\begin{equation}\label{eq:MAq-filter-recursion}
 \psi_0=1,
 \qquad
 \psi_k=-\sum_{\ell=1}^{\min(q,k)}\theta_\ell\psi_{k-\ell},
 \quad k=1,\ldots,n-1,
\end{equation}
and set $\psi_k=0$ for $k<0$ and $\theta_k=0$ for $k>q$.
Let $F\in\R^{n\times n}$ and $H\in\R^{n\times q}$ have entries
$F_{ij}=\psi_{i-j}$ and $H_{ia}=\theta_{i+a-1}$, and put
$W=F'F$ and $M=I_q+H'WH$.

\begin{proposition}[BCP coordinates of a moving-average correlation matrix]
\label{prop:ma-correlation}
The matrix in \eqref{eq:MAq-correlations} is positive definite, and its
precision matrix and balancing weights satisfy
\begin{equation}\label{eq:MAq-precision}
 K=s\left(W-WHM^{-1}H'W\right),
 \qquad \delta_i=K_{ii}^{1/4}.
\end{equation}
For $i\ne j$,
\begin{equation}\label{eq:BFT-MAq}
 [B(C)]_{ij}
 =\frac12\left(
 \delta_i\delta_j\rho_{|i-j|}
 -\frac{K_{ij}}{\delta_i\delta_j}
 \right).
\end{equation}
In particular, $[B(C)]_{ij}=-K_{ij}/(2\delta_i\delta_j)$ whenever
$|i-j|>q$.
\end{proposition}

\Needspace{2\baselineskip}
The only matrix inverse in \eqref{eq:MAq-precision} has dimension $q$;
the formula does not require invertibility of the MA process.
For MA(1), $\rho_1=\theta_1/(1+\theta_1^2)$ and $C$ is tridiagonal.
More generally, although every correlation beyond lag $q$ is zero,
the corresponding precision entries are generally nonzero. The BCP
matrix is therefore generally dense, despite the finite range of the
marginal correlations.

\subsection{Autoregressive correlation matrices}
\label{subsec:ar-correlations}

Consider a zero-mean, stationary causal AR($p$) process with
unit marginal variance,
\begin{equation*}
 x_t=\sum_{\ell=1}^p\phi_\ell x_{t-\ell}+\varepsilon_t,
 \qquad \operatorname{Var}(\varepsilon_t)=\sigma^2>0,
\end{equation*}
where $\varepsilon_t$ is zero-mean white noise and
$1-\sum_{\ell=1}^p\phi_\ell z^\ell$ has all its roots outside
the unit circle. For $n>p\ge1$, its correlation matrix is Toeplitz,
$C_{ij}=\rho_{|i-j|}$, where $\rho_0=1$,
$\rho_h=\sum_{\ell=1}^p\phi_\ell\rho_{|h-\ell|}$ for $h\ge1$, and
$\sigma^2=1-\sum_{\ell=1}^p\phi_\ell\rho_\ell$.
Its precision matrix is banded.
Let $C_p=(\rho_{|i-j|})_{i,j=1}^p$, let
$E_p=(I_p,0)\in\R^{p\times n}$, and define
$r_s=e_s-\sum_{\ell=1}^p\phi_\ell e_{s-\ell}$ for $s=p+1,\ldots,n$,
where $e_s$ is the $s$th coordinate vector in $\R^n$.

\begin{proposition}[BCP coordinates of an autoregressive correlation matrix]
\label{prop:ar-correlation}
The precision matrix is
\begin{equation}\label{eq:ARp-precision}
 K=E_p'C_p^{-1}E_p+\frac{1}{\sigma^2}\sum_{s=p+1}^n r_s r_s',
 \qquad K_{ij}=0\quad\text{if }|i-j|>p.
\end{equation}
With $\delta_i=K_{ii}^{1/4}$ and $h=|i-j|$, the off-diagonal BCP entries are
\begin{equation}\label{eq:BFT-ARp}
 [B(C)]_{ij}
 =\begin{cases}
 \dfrac12\left(\delta_i\delta_j\rho_h
       -\dfrac{K_{ij}}{\delta_i\delta_j}\right),&1\le h\le p,\\[8pt]
 \dfrac12\delta_i\delta_j\rho_h,&h>p.
 \end{cases}
\end{equation}
\end{proposition}

Equation~\eqref{eq:ARp-precision} includes the finite-sample boundary
corrections and requires only a $p\times p$ inverse. Beyond lag $p$,
each BCP coordinate is a positive rescaling of the corresponding
correlation, so its sign and zero status are preserved exactly.
For fixed $p$, only $O(np)$ pairs receive a precision correction.
Scalar formulas for the interior entries are given in
\eqref{eq:ARp-interior-BFT} of the supplement.

\Needspace{5\baselineskip}
When $p=1$ and $\phi_1=\rho$, $C$ coincides with the correlation matrix
of the homogeneous Gaussian path in \Cref{subsec:tree-correlations}.
In this case, $K_{11}=K_{nn}=1/(1-\rho^2)$ and
$K_{ii}=(1+\rho^2)/(1-\rho^2)$ for $2\le i\le n-1$.
With $\delta_i=K_{ii}^{1/4}$,
\begin{equation}\label{eq:BFT-AR1}
 [B(C)]_{ij}
 =\frac12\delta_i\delta_j\rho^{|i-j|}
 +\frac{\rho\,\indic\{|i-j|=1\}}
 {2(1-\rho^2)\delta_i\delta_j},
 \qquad i\ne j.
\end{equation}
The endpoint and interior diagonal values differ, so the transformed
matrix need not be Toeplitz. Gaussianity is not needed for the correlation
or BCP formulas above.

\subsection{Standard factor-model correlation matrices}

Consider standardized variables with $r$ orthogonal common factors.  Let the
$i$th row of the $n\times r$ loading matrix $\Lambda$ be $\lambda_i'$ and
set
\begin{equation}\label{eq:factor-correlation}
 C=Q+\Lambda\Lambda',
 \qquad
 Q=\diag(q_1,\ldots,q_n),
 \qquad
 q_i=1-\lambda_i'\lambda_i>0.
\end{equation}
Then $C$ is a correlation matrix.  Define $M$ and $\delta_i>0$ by
\begin{equation}\label{eq:factor-M-delta}
 M=I_r+\Lambda'Q^{-1}\Lambda,
 \qquad
 \delta_i^4=\frac1{q_i}
 -\frac{\lambda_i'M^{-1}\lambda_i}{q_i^2}.
\end{equation}

\begin{proposition}[BCP coordinates of a factor-model correlation matrix]
\label{prop:factor-correlation}
For the matrix in \eqref{eq:factor-correlation} and $i\ne j$,
\begin{equation}\label{eq:BFT-factor-correlation}
 [B(C)]_{ij}
 =\frac12\left\{
  \delta_i\delta_j\lambda_i'\lambda_j
  +\frac{\lambda_i'M^{-1}\lambda_j}
  {q_iq_j\delta_i\delta_j}
 \right\}.
\end{equation}
Equivalently, if
\begin{equation}\label{eq:factor-dual-loadings}
 u_i=\delta_i\lambda_i,
 \qquad
 v_i=\frac{M^{-1/2}\lambda_i}{q_i\delta_i},
\end{equation}
where $M^{-1/2}$ is the symmetric positive-definite inverse square root, then
\begin{equation}\label{eq:BFT-factor-Gram}
 [B(C)]_{ij}
 =\frac12\left(u_i'u_j+v_i'v_j\right),
 \qquad i\ne j.
\end{equation}
\end{proposition}

Thus each off-diagonal BCP coordinate averages two loading inner products:
one formed from the risk-scaled original loadings and the other from the
precision-dual loadings. Evaluation uses $r\times r$ matrix operations
and pairwise inner products, avoiding an $n\times n$ eigendecomposition.
These finite-dimensional formulas do not require pervasive factors; the
additional conditions in \Cref{ass:sparse-factor} yield the asymptotic
sparsity result in \Cref{thm:sparse-bft}.

For a one-factor model, write $\Lambda=\beta$, $q_i=1-\beta_i^2$, and
$\kappa=1+\sum_{l=1}^n\beta_l^2/q_l$.
Then $\delta_i^4=q_i^{-1}-\beta_i^2/(\kappa q_i^2)$ and $
 [B(C)]_{ij}
 =\frac{\beta_i\beta_j}{2}
 \left(
  \delta_i\delta_j+
  \frac{1}{\kappa q_iq_j\delta_i\delta_j}
 \right)$ for $i\ne j$.

\section{Sparsity preservation in high dimensions}
\label{sec:high-dimensional-sparsity}

Equicorrelation provides a simple illustration of how transformed coordinates
behave as the dimension increases. For a fixed correlation $0<\rho<1$,
\Cref{prop:equicorrelation} gives
\begin{equation*}
 b_n(\rho)\longrightarrow\frac{\rho}{2\sqrt{1-\rho}}>0,
 \qquad
 g_n(\rho)=\frac1n\log\frac{1+(n-1)\rho}{1-\rho}\longrightarrow0.
\end{equation*}
The off-diagonal LS coordinates also equal $g_n(\rho)$. Thus BCP retains a
nonvanishing coordinate for a fixed nonzero correlation, whereas GFT and LS
coordinates shrink. This example motivates the question of whether vanishing
and nonvanishing correlations remain distinguishable in transformed coordinates.

The factor-model formulas in \Cref{prop:factor-correlation} provide a
basis for addressing this question.
Cholesky and spherical parametrizations have also been developed for
estimation with a given pattern of zero correlations
\citep{SaxenaBhatPinjari:2023,HuangYeWang:2025}. Here we study the asymptotic
correspondence between marginal correlations and unrestricted BCP
coordinates under suitable factor conditions. Each BCP coordinate then
approximates a positively rescaled correlation uniformly over all pairs,
keeping weak marginal associations small without driving material
associations to zero. This correspondence permits support recovery under
suitable signal and estimation-error conditions, although setting a BCP
coordinate to zero need not produce an exact zero correlation at finite
dimension. All proofs and detailed comparisons are collected in
Supplementary \Cref{app:sparsity-proofs}.

\Needspace{5\baselineskip}
\subsection{Factor strength and residual variation}

\begin{assumption}[Increasing factor strength with diagonal residual covariance]
\label{ass:sparse-factor}
For each $n$, let
\begin{equation*}
 C_n=\Lambda_n\Lambda_n'+\Omega_n\in\CorP{n},\qquad
 \Omega_n=\diag(\omega_{1,n},\ldots,\omega_{n,n}).
\end{equation*}
The matrix $\Lambda_n$ has $r_n\geq1$ linearly independent columns.
For a constant $a>0$ independent of $n$,
\begin{equation}\label{eq:sparse-factor-bounds}
 a\leq\omega_{i,n}\leq1,\qquad
 \mu_n:=\lambda_{\min}(\Lambda_n'\Omega_n^{-1}\Lambda_n)
 \rightarrow\infty.
\end{equation}
The transpose of row $i$ of $\Lambda_n$ is denoted by $\lambda_{i,n}$.
\end{assumption}

Here $\omega_{i,n}$ measures residual variance after removing the common
factors; $\tau_i^2$ in \Cref{prop:partial-correlation} measures residual
variance after projection on the remaining observed variables. The diagonal
residual covariance requires the residuals to be mutually uncorrelated,
while the lower bound on their variances prevents any variable from being
almost entirely explained by the factors. Since $C_n$
has unit diagonal, the loadings satisfy
$\norm{\lambda_{i,n}}^2=1-\omega_{i,n}\leq1-a$.
To interpret the factor-strength condition, write
\begin{equation*}
 \mu_n
 =\min_{\norm{v}_2=1}v'\Lambda_n'\Omega_n^{-1}\Lambda_n v
 =\min_{\norm{v}_2=1}\sum_{i=1}^n
 \frac{(\lambda_{i,n}'v)^2}{\omega_{i,n}}.
\end{equation*}
Here $v\in\mathbb R^{r_n}$ selects a unit direction in factor space,
corresponding to a linear combination of the factors, and
$\lambda_{i,n}'v$ is variable $i$'s loading in that direction. Dividing
its squared loading by $\omega_{i,n}$ measures its contribution relative
to residual noise, giving greater weight to variables with smaller
residual variances. Summing across variables gives the aggregate strength
of that direction, and taking the minimum identifies the weakest direction.

Thus, $\mu_n\to\infty$ requires factor strength to diverge uniformly
across all unit directions as the panel grows. Increasing dimension alone
is insufficient: merely adding variables unrelated to the weakest direction
does not strengthen it. The condition does not require each direction to
affect a fixed fraction of the panel. The number of factors may also
increase, although
\begin{equation*}
 r_n\mu_n\leq\tr(\Lambda_n'\Omega_n^{-1}\Lambda_n)
 \leq\frac{n(1-a)}{a}
\end{equation*}
implies $r_n=o(n)$.

A familiar special case is a fixed number of pervasive factors with
$\mu_n\geq\kappa n$ for some $\kappa>0$.
The assumption also allows semi-strong factors, for example
$\mu_n\asymp n^\alpha$ with $0<\alpha<1$
\citep{ChudikPesaranTosetti:2011}.
With one factor per disjoint group and all within-group loading magnitudes
uniformly bounded away from zero, $\mu_n\asymp\min_k n_k$, where $n_k$
is the size of group $k$. The condition therefore holds when the smallest
group grows.
More general loading patterns must satisfy the same eigenvalue condition.
The assumption permits heterogeneous loadings and residual variances and
does not require the proportion of nonzero correlations to vanish.

Here sparsity concerns marginal uncorrelatedness, $C_{ij,n}=0$, rather
than zeros of the precision matrix. For Gaussian vectors, covariance zeros
encode marginal independence \citep{ChaudhuriDrtonRichardson:2007}, whereas
precision zeros encode conditional independence given the remaining variables
\citep{Dempster:1972}. Related screening methods for correlation and
partial-correlation hubs are studied by \citet{HeroRajaratnam:2012}.

\subsection{Entrywise preservation by BCP}

Let $K_n=C_n^{-1}$, $\delta_{i,n}=K_{ii,n}^{1/4}$, and
$B_n=B(C_n)$. The exact coordinate formula is
\begin{equation}\label{eq:sparse-bft-exact}
 B_{ij,n}=\frac12\left(
 \delta_{i,n}\delta_{j,n}C_{ij,n}
 -\frac{K_{ij,n}}{\delta_{i,n}\delta_{j,n}}\right),\qquad i\ne j.
\end{equation}
The weakest factor strength also controls the approximation error.
Under \Cref{ass:sparse-factor}, Woodbury's identity expresses $K_n$ as
the diagonal matrix $\Omega_n^{-1}$ minus a factor correction involving
$(I_{r_n}+\Lambda_n'\Omega_n^{-1}\Lambda_n)^{-1}$. Since
\begin{equation*}
 \left\|(I_{r_n}+\Lambda_n'\Omega_n^{-1}\Lambda_n)^{-1}\right\|_2
 =\frac{1}{1+\mu_n},
\end{equation*}
the bounds on loadings and residual variances imply that the off-diagonal
precision entries are uniformly $O(\mu_n^{-1})$, while the diagonal
scaling remains bounded. The following theorem translates these bounds
into a uniform approximation of BCP coordinates by positively rescaled
correlations.

\begin{theorem}[Uniform preservation of the correlation entries]
\label{thm:sparse-bft}
Under \Cref{ass:sparse-factor}, set
$h_n=\{a^2(1+\mu_n)\}^{-1}$ and
$K_a=\tfrac12+\tfrac14a^{-1/4}$. Then
\begin{equation}\label{eq:sparse-bft-expansion}
 B_{ij,n}=s_{ij,n}C_{ij,n}+R_{ij,n},\qquad
 s_{ij,n}=\frac{1}{2(\omega_{i,n}\omega_{j,n})^{1/4}},
\end{equation}
where, uniformly over $i\ne j$,
\begin{equation}\label{eq:sparse-bft-error}
 \frac12\leq s_{ij,n}\leq\frac{1}{2\sqrt a},\qquad
 |R_{ij,n}|\leq K_a h_n=O(\mu_n^{-1}).
\end{equation}
At an exact zero, the sharper bound is
\begin{equation}\label{eq:sparse-null-bound}
 C_{ij,n}=0\quad\Rightarrow\quad
 |B_{ij,n}|\leq\frac{h_n}{2}.
\end{equation}
\end{theorem}

Because the scaling is positive and bounded above and away from zero,
\eqref{eq:sparse-bft-expansion} gives, uniformly over $i\ne j$,
\begin{equation}\label{eq:sparse-vanishing-equivalence}
 |B_{ij,n}|\leq\frac{|C_{ij,n}|}{2\sqrt a}+K_ah_n,
 \qquad
 |C_{ij,n}|\leq2|B_{ij,n}|+2K_ah_n,
\end{equation}
with $h_n\rightarrow0$. Thus, for a fixed pair, or along any sequence of
pairs $(i_n,j_n)$, the coordinate $B_{i_nj_n,n}$ vanishes as
$n\rightarrow\infty$ if and only if the correlation $C_{i_nj_n,n}$ does.
A correlation bounded away from zero retains a nonvanishing
coordinate, while a zero or vanishing correlation has a vanishing
coordinate. If $\mu_n|C_{i_nj_n,n}|\to\infty$, the relative approximation
error vanishes and the signs agree eventually. The uniform entrywise
bound also allows statements about all pairs simultaneously; it does not
assert a vanishing operator-norm error. This preservation concerns the limit $n\to\infty$. At finite dimension,
\eqref{eq:sparse-bft-exact} gives $B_{ij,n}=0$ at a correlation zero
only when $K_{ij,n}=0$ as well. Increasing dimension alone is insufficient;
Supplementary \Cref{rem:sparse-exact-zeros} gives a fixed-block example.
For pervasive factors, the remainder is $O(n^{-1})$; if
$\mu_n\asymp n^\alpha$, it is $O(n^{-\alpha})$.

The factor setting provides an interpretable sufficient regime for this
equivalence. More generally, it suffices that the precision diagonals are
uniformly bounded and the off-diagonal precision entries vanish uniformly:
$\max_i K_{ii,n}\le M<\infty$, $\varepsilon_n:=\max_{i\ne j}|K_{ij,n}|\rightarrow0$. The first condition means that no standardized variable becomes almost
perfectly linearly predictable from all the others: its prediction residual
variance is $K_{ii,n}^{-1}\ge M^{-1}$. Under this bound, the second condition
is equivalent to uniformly vanishing full partial correlations.
Indeed, since $K_{ii,n}\ge1$, \eqref{eq:sparse-bft-exact} expresses
$B_{ij,n}$ as $\tfrac12(K_{ii,n}K_{jj,n})^{1/4}C_{ij,n}$ plus an error
bounded by $\varepsilon_n/2$, with the multiplier between $1/2$ and
$\sqrt M/2$. Thus the same vanishing-coordinate equivalence holds.
\Cref{ass:sparse-factor} implies these conditions with $M=a^{-1}$ and
$\varepsilon_n\le h_n=O(\mu_n^{-1})$; we retain that setting for the explicit
rates and support results below.

\subsection{Support recovery}

Thresholding is widely used for sparse covariance estimation
\citep{BickelLevina:2008}. Related work studies large-scale screening for
variables with strong correlations
\citep{HeroRajaratnam:2011,WeiRajaratnamHero:2023}. Here a common magnitude threshold on BCP
coordinates can distinguish null and nonnull pairs when the smallest nonzero
correlation dominates the approximation error.

\begin{corollary}[Recovery of the entire support]
\label{cor:sparse-support}
Under \Cref{ass:sparse-factor}, let
$\mathcal S_n=\{(i,j):i<j,\ C_{ij,n}\ne0\}$ and
$m_n=\min_{(i,j)\in\mathcal S_n}|C_{ij,n}|$. If
$\mu_n^{-1}=o(\tau_n)$ and $\tau_n=o(m_n)$, then, for all sufficiently large
$n$,
\begin{equation*}
 \{(i,j):i<j,\ |B_{ij,n}|>\tau_n\}=\mathcal S_n.
\end{equation*}
The retained coordinates also have the signs of the corresponding
correlations. At any finite $n$, it is sufficient to choose
\begin{equation*}
 h_n/2<\tau_n<m_n/2-K_a h_n,
\end{equation*}
whenever this interval is nonempty. These uniform finite-dimension bounds
are sufficient and can be conservative.
\end{corollary}

The condition concerns the entire support, not just a fixed pair.
For example, if all nonzero correlations are bounded away from zero,
$\tau_n=\mu_n^{-1/2}$ is admissible. If the support is empty, only the null
bound is needed. A separating threshold exists when $\mu_nm_n\rightarrow\infty$.
For estimated coordinates, support recovery also requires estimation error
to be smaller than the threshold.

\begin{corollary}[Support recovery from estimated coordinates]
\label{cor:sparse-estimated-support}
Under \Cref{ass:sparse-factor}, let $m_n$ be the smallest nonzero
correlation magnitude and let $B_n=B(C_n)$. Suppose symmetric hollow
estimates satisfy
\begin{equation*}
 \max_{i\ne j}|\hat B_{ij,n}-c_nB_{ij,n}|=O_p(e_n)
\end{equation*}
for positive scalars $c_n$, possibly random, such that $c_n$ and
$c_n^{-1}$ are $O_p(1)$. If
\begin{equation*}
 \mu_n^{-1}+e_n=o(\tau_n),\qquad \tau_n=o(m_n),
\end{equation*}
then thresholding $\hat B_n$ at $\tau_n$ recovers $\mathcal S_n$
and the signs of its correlations with probability tending to one as $n\to\infty$.
\end{corollary}

The error rate $e_n$ depends on the estimator. The factor $c_n$ allows
common multiplicative bias without changing the asymptotic separation
orders; a particular finite-sample threshold still depends on scale.
For independent Gaussian observations with $n,T\to\infty$ and $n/T$
bounded away from one from below,
Supplementary \Cref{prop:sparse-plugin} establishes this condition for
$\hat B_n=B(\hat C_n)$ with
$c_n=(1-n/T)^{-1/2}$ and $e_n=\sqrt{\log n/T}$.
The result accommodates a nonvanishing dimension-to-sample-size ratio.

\subsection{Reconstruction after imposing zero coordinates}

Support recovery concerns the forward coordinates. For modeling, it is
also useful to know whether imposing their zeros retains the original
correlations after inversion. The following result gives such a guarantee
with a bounded number of factors and allows arbitrary heterogeneous
loadings subject to \Cref{ass:sparse-factor}.

\begin{proposition}[Reconstruction from the correlation support]
\label{prop:sparse-reconstruction}
Suppose \Cref{ass:sparse-factor} holds and $\sup_n r_n<\infty$.
Let $\tilde B_n$ equal $B_n=B(C_n)$ except that its entries are
set to zero wherever $C_{ij,n}=0$. Then
\begin{equation*}
 \max_{i,j}\bigl|[C(\tilde B_n)]_{ij}-C_{ij,n}\bigr|
 =O(\mu_n^{-1}).
\end{equation*}
In the pervasive case $\mu_n\geq\kappa n$, this error is $O(n^{-1})$.
\end{proposition}

In the pervasive case, the trace bound
$r_n\mu_n\leq n(1-a)/a$ gives
$r_n\leq(1-a)/(a\kappa)$, so the bounded-factor condition is automatic.
The result also covers a bounded number of semi-strong factors, with the
rate determined by the weakest factor direction.

The proposition concerns zeroing coordinates at population correlation
zeros and retaining the remaining coordinates unchanged. It controls the
entire reconstructed matrix, including its originally zero entries.
Because the zero locations are supplied, no minimum-signal condition is
needed; identifying them from data requires separate support-recovery
conditions. The reconstructed matrix need not retain the original factor
representation. Positive definiteness follows from the inverse map at
every dimension.

\subsection{Comparison with the GFT, InvC, and LS methods}
\label{sec:sparsity-comparators}

What separates BCP from GFT, InvC, and LS is the relation between a
coordinate and the corresponding marginal correlation. Write $H_n=\Off(\log C_n)$ for GFT;
InvC uses the off-diagonal entries of $K_n=C_n^{-1}$.
Under \Cref{ass:sparse-factor},
\begin{equation*}
 \max_{i\ne j}|H_{ij,n}|=O(\log(1+\mu_n)/\mu_n),\qquad
 \max_{i\ne j}|K_{ij,n}|\leq h_n=O(\mu_n^{-1}),
\end{equation*}
including at correlations bounded away from zero. The InvC bound follows
directly from Woodbury's identity. In the pervasive special case $\mu_n\geq\kappa n$ with homogeneous
residual variance $\Omega_n=\omega I_n$, let $Q_n=\Lambda_n'\Lambda_n/n$.
The leading terms are
\begin{equation*}
 H_{ij,n} =\frac{\log n}{n}\lambda_{i,n}'Q_n^{-1}\lambda_{j,n}
            +O(n^{-1}), \qquad  K_{ij,n} =-\frac{1}{\omega n}\lambda_{i,n}'Q_n^{-1}\lambda_{j,n}
            +O(n^{-2}),
\end{equation*}
whereas $C_{ij,n}=\lambda_{i,n}'\lambda_{j,n}$ for $i\ne j$.
Thus both leading terms depend on a different inner product from the
correlation itself, and their zeros need not agree with correlation zeros.

For LS, write $S_n=\LogLS(C_n)=\log(D_nC_nD_n)$, where the positive
diagonal matrix $D_n$ satisfies $D_nC_nD_n\one_n=\one_n$.
If its largest-to-smallest diagonal
ratio is uniformly bounded, then
$\max_{i\ne j}|S_{ij,n}|=O(\log(1+\mu_n)/\mu_n)$ as well. This extra scaling
condition does not follow from \Cref{ass:sparse-factor}.
The equicorrelation example at the beginning of this section
illustrates these contrasting limits explicitly.

The fact that all GFT coordinates vanish does not itself preclude support
recovery, because coordinates associated with zero correlations could
vanish faster than those associated with nonzero correlations. However,
the factor conditions do not guarantee this separation. The counterexample
in Supplementary \Cref{ex:sparse-ah-failure} has a zero correlation whose
GFT coordinate has magnitude of order $(\log n)/n$, whereas a fixed nonzero
correlation has a coordinate of order $n^{-1}$. The zero-correlation
coordinate therefore vanishes more slowly, preventing a common magnitude
threshold from separating the two pairs, even after common rescaling.

The same example establishes failure of threshold separation for InvC,
while Supplementary \Cref{ex:sparse-ls-failure} establishes it for LS
within the factor class.
\Cref{prop:sparse-ah-shrinkage,prop:sparse-ah-leading,prop:sparse-ls-regular}
give the GFT and LS formulas and their conditions. Thus GFT, InvC, and LS
do not share BCP's general coordinate-thresholding guarantee.
All four maps remain
invertible and preserve block-diagonal zeros; the comparison concerns
entrywise marginal support, not loss of information or the exclusion of
special subclasses.

\section{Application to high-dimensional equity returns}
\label{sec:stock-residual-sparsity}

We compare BCP, GFT, LS, and InvC using a high-dimensional panel of
idiosyncratic daily stock returns. We choose GFT, LS, and InvC as benchmarks
because, like BCP, they are permutation-equivariant, so the comparison does
not depend on the ordering of the stocks. We first remove the Fama--French
five factors (FF5) and momentum, then standardize each stock's residuals by
its fitted conditional volatility. Our focus is on how the transformed
coordinates represent the approximate sparsity pattern and entry
magnitudes of the resulting shock correlation matrix.

\subsection{Data, factor adjustment, and volatility standardization}

We use daily equity price returns from CRSP for the period from
January 3, 2011 through December 31, 2025. The six predictors are the
market excess return, SMB, HML, RMW, and CMA
from FF5 \citep{Fama_French_2015}, together with momentum. We retain all securities with finite returns on all $T=3{,}772$ trading
dates. Of 8,670 candidates, 1,765 form a balanced panel and 6,905 are
excluded for incomplete returns. The requirement of a complete 15-year
history selects continuously observed securities, to which the empirical
findings pertain. For each security $i$, we estimate an ordinary least-squares
regression over the full period,
\begin{equation*}
 r_{it}=\mu_i+\vartheta_i' f_t+u_{it},\qquad
 f_t=(\mathrm{MKT}_t,\mathrm{SMB}_t,\mathrm{HML}_t,
       \mathrm{RMW}_t,\mathrm{CMA}_t,\mathrm{Mom}_t)'.
\end{equation*}
For each stock, we fit a zero-mean GJR-GARCH(1,1) model to the OLS
residuals \citep{GlostenJagannathanRunkle:1993}. Writing $u_{it}$ and $h_{it}$ for the fitted
residuals and conditional variances, we form standardized shocks
$z_{it}=u_{it}/\sqrt{h_{it}}$ and compute their sample
Pearson correlation matrix $\hat C$. This baseline is positive definite
and uses no shrinkage or entrywise thresholding. We compute the BCP matrix
$\hat B=B(\hat C)$ and the GFT, LS, and InvC coordinates from the
same $\hat C$. Supplementary \Cref{app:empirical-implementation}
gives estimation and reconstruction details.

\subsection{Visual comparison of transformed coordinates}

In \Cref{fig:stock-residual-heatmaps}, stocks are ordered by their final
available eight-digit GICS subindustry code within the estimation
period. Every panel uses the same fixed ordering; strips and boundaries
identify the broader GICS sectors. The standardized-shock correlation matrix
contains broad regions of weak dependence alongside localized positive
blocks and negative bands. BCP retains much of this structure, including
the signs and relative magnitudes of the entries. GFT appears
substantially paler on the shared scale. The visual pattern of LS is very
similar to GFT and is omitted from this figure. InvC displays a different
pattern of linear partial association, with prominent negative entries in
several positively correlated blocks.

\begin{figure}[htp]
 \centering
 \includegraphics[width=\textwidth]{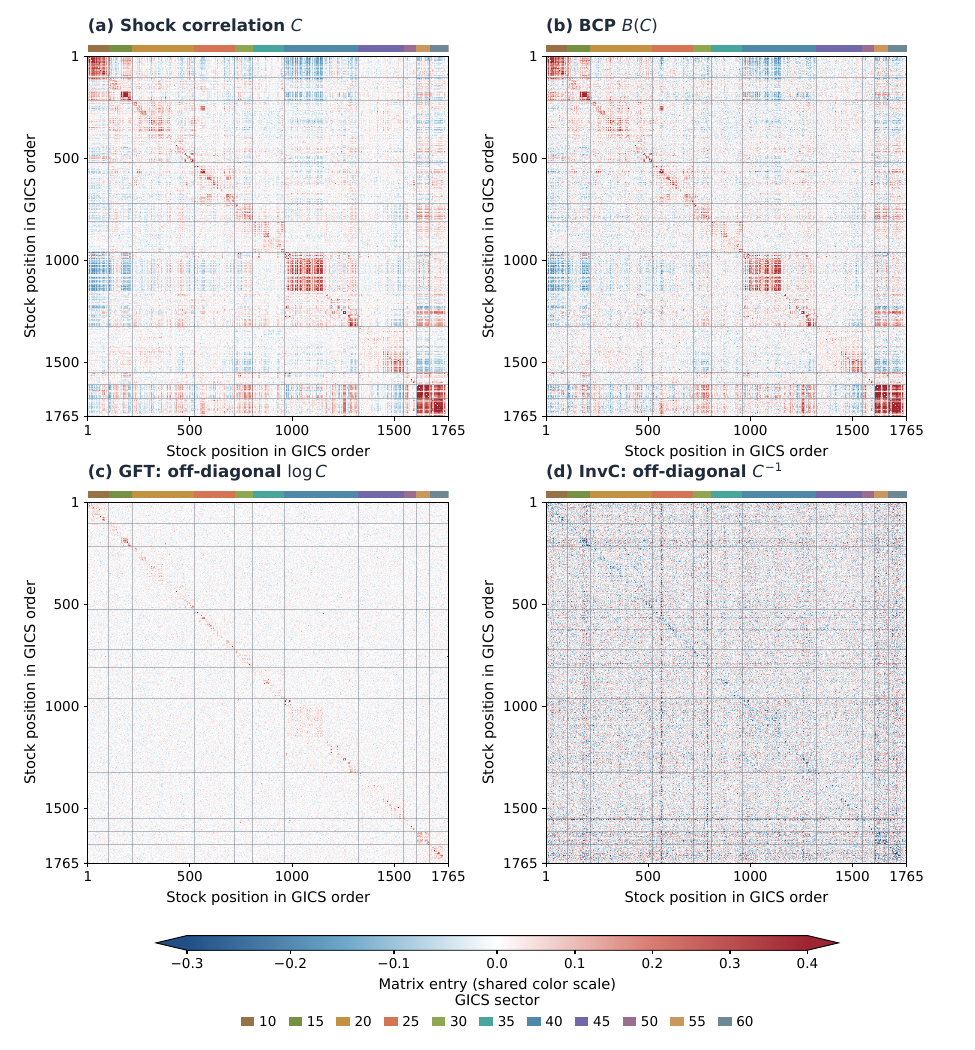}
 \caption{Standardized-shock correlations and transformed coordinates for
 1,765 equities, 2011--2025. Each stock's raw daily price return is
 regressed on an intercept, FF5, and momentum over 3,772 common dates;
 its residuals are then standardized using a fitted GJR-GARCH(1,1) model.
 Panels show the shock correlation matrix, BCP, GFT, and InvC in the
 same GICS order. The shared color scale is $[-0.3,0.4]$, centered at
 white for zero, with saturated tails. All four panel diagonals are masked
 only for display. No matrix entries are thresholded. The visual pattern
 of LS, omitted here, is very similar to GFT.}
 \label{fig:stock-residual-heatmaps}
\end{figure}

\Cref{fig:stock-smoothed-coordinates} complements the heatmaps by
summarizing transformed coordinates conditional on standardized-shock correlations.
We estimate the median and the 5th and 95th percentiles using Gaussian
weights with bandwidth 0.01, truncated at three bandwidths. The shaded
band represents the middle 90\% of coordinates, showing their dispersion
at each correlation value. Over the displayed range, the BCP conditional median closely follows the
45-degree reference, tracking both the sign and magnitude of shock
correlations much more closely than the other medians.
The GFT and LS medians are much flatter and closer to the zero-coordinate
line, particularly over the negative-correlation range, where they remain
nearly constant even as correlations move farther from zero. InvC has
a relatively flat, downward-sloping median and a wider band: its
coordinates encode scaled negative partial correlations, rather than
marginal correlations. Small coordinates under these alternatives can
therefore coexist with material marginal correlations.

\begin{figure}[htbp]
 \centering
 \setlength{\abovecaptionskip}{4pt}
 \includegraphics[width=\textwidth]{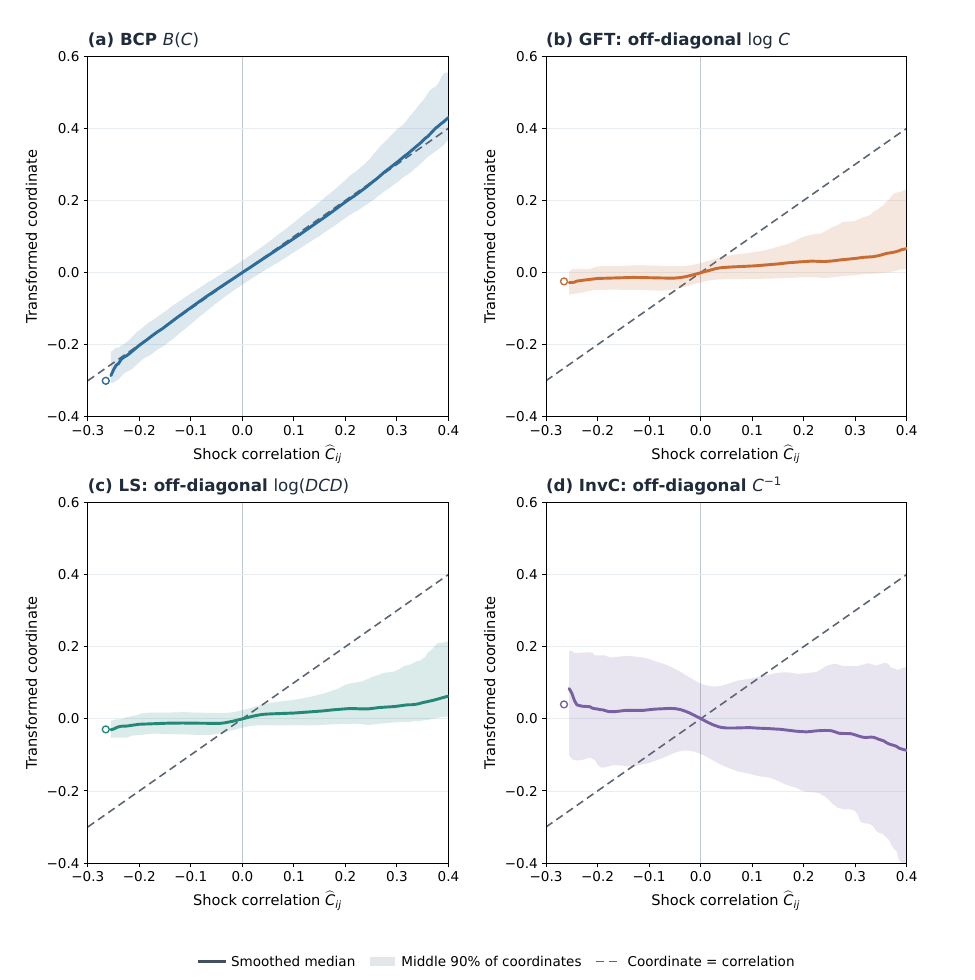}
 \caption{Transformed coordinates conditional on standardized-shock correlations,
 2011--2025. The sample correlation matrix $\hat C$ has dimension
 $1{,}765\times1{,}765$, giving $1{,}556{,}730$ unique off-diagonal pairs $i<j$.
 The horizontal axis shows the shock correlations $\hat C_{ij}$, and the
 vertical axis shows the corresponding off-diagonal coordinates under BCP,
 GFT, LS, and InvC in panels (a)--(d), respectively. Solid curves are
 Gaussian-kernel weighted medians; shading spans the corresponding
 5th--95th percentiles of coordinates, not confidence intervals.
 The dashed line is the 45-degree reference. Hollow markers show
 original bin medians in sparsely populated tails. The displayed range
 excludes 3,097 pairs (0.20\%), all with correlations above 0.4.}
 \label{fig:stock-smoothed-coordinates}
\end{figure}

\subsection{Numerical comparison of sparse reconstructions}
\label{sec:stock-sparse-reconstructions}

We next examine whether this close correspondence translates into accurate
correlation reconstructions when a common set of coordinates is set to zero.
We select this set using two-sided
5\% tests of the original sample correlations. Standard errors use their
estimated influence functions, accounting for sample centering and scaling
of the fitted shocks, with the Bartlett HAC estimator of \citet{NeweyWest:1987} and nine lags.
The OLS and volatility filters are treated as fixed in these tests.
The tests are pointwise, without a multiplicity adjustment; an insignificant
estimate does not establish a population zero. For each insignificant
$\hat C_{ij}$, we set the corresponding coordinate to zero in all four
transformations and recover a correlation matrix through each inverse map.
For BCP, we set
$\tilde C_{\mathrm{BCP}}=C(\tilde B)$, where the sparse
coordinate matrix is defined by
\begin{equation*}
 \tilde B_{ij}
  =\hat B_{ij}\,
 \indic\!\left\{
 |\hat C_{ij}|>
 z_{0.975}\times {\operatorname{se}}(\hat C_{ij})
 \right\},\qquad i\ne j,
\end{equation*}
where $z_{0.975}$ is the standard normal 97.5th percentile. The same
correlation-based indicator selects the retained GFT, LS, and InvC
coordinates. All four methods therefore set the same 67.89\% of coordinates
to zero. The comparison isolates how these common zero restrictions
translate through the inverse maps, rather than selecting a separate
specification for each method.

Let $\hat{C}_{0.05}$ denote the reference obtained by setting insignificant
entries of $\hat C$ directly to zero. Although $\hat C$ is positive
definite, direct thresholding need not preserve this property
\citep{BickelLevina:2008,GuillotRajaratnam:2012}. Here it produces an indefinite matrix, with
smallest eigenvalue $-0.6356$. In contrast, each of the four inverse maps
yields a positive-definite correlation matrix after coordinate thresholding.
The relevant question is therefore how closely these valid reconstructions
preserve the original correlations and approximate the desired sparse
pattern. Exact zeros in coordinates need not yield exact correlation
zeros, and $\hat{C}_{0.05}$ serves only as an entrywise reference.

\begin{table}[htbp]
\centering
\caption{Correlation reconstruction under common coordinate-zero restrictions.}
\label{tab:stock-common-c-mask}
\begingroup\small
\setlength{\tabcolsep}{3pt}
\begin{tabular*}{\textwidth}{@{\extracolsep{\fill}}p{.38\textwidth}ccccc@{}}
\toprule
Statistic & Original & BCP & GFT & LS & InvC \\
\midrule
Retained coordinates $k$ & 1,556,730 & 499,883 & 499,883 & 499,883 & 499,883 \\
Coordinates set to zero &  & 67.89\% & 67.89\% & 67.89\% & 67.89\% \\
Pearson correlation with $\hat C$ & 1.0000 & 0.9747 & 0.6985 & 0.6688 & 0.4758 \\
Pearson correlation with $\hat{C}_{0.05}$ & 0.9624 & 0.9883 & 0.6532 & 0.6275 & 0.4218 \\
Off-diagonal RMSE versus $\hat C$ & 0.0000 & 0.0125 & 0.1978 & 0.1856 & 0.6126 \\
Off-diagonal RMSE versus $\hat{C}_{0.05}$ & 0.0150 & 0.0087 & 0.2017 & 0.1892 & 0.6164 \\
Gaussian log-likelihood $(\times 10^3)$ & -7,520 & -8,332 & -8,740 & -9,075 & -20,619 \\
BIC$^{*}\ (\times 10^3)$ & 27,859 & \textbf{20,781} & 21,596 & 22,268 & 45,356 \\
\bottomrule
\end{tabular*}
\endgroup
\par\smallskip\begin{minipage}{\textwidth}\footnotesize
\emph{Notes:} Original denotes the sample correlation $\hat C$ of the GJR-GARCH standardized shocks for 2011--2025. The other columns report reconstructions under the common 5\% restrictions. Pearson correlations and RMSE use all unique off-diagonal pairs, including zeros; $\hat{C}_{0.05}$ is the directly thresholded reference.
$^{*}$BIC is the descriptive plug-in score $-2\ell+k\log T$, with $T=3{,}772$; the sparse estimates are not maximum-likelihood refits. Log-likelihood and BIC are in thousands; the lowest BIC is bold.
\end{minipage}
\end{table}

\Needspace{2\baselineskip}
\Cref{tab:stock-common-c-mask} shows that BCP reconstructs the original
correlations and approximates the thresholded reference most closely.
Among the four sparse reconstructions, it has the highest off-diagonal
Pearson correlation and the lowest RMSE relative to both references.
Pearson correlation measures agreement in the cross-sectional pattern,
whereas RMSE also captures differences in entry magnitudes. BCP is closer
to $\hat{C}_{0.05}$ in RMSE than the original $\hat C$ is, so the
restrictions move the estimate toward the desired sparse pattern while
preserving positive definiteness. GFT and LS depart more from both
references, and InvC departs the most. Thus, although all four inverse
maps yield valid correlation matrices, identical coordinate zeros have
markedly different effects on marginal correlations. BCP's closer
reconstruction is consistent with the closer correspondence in
\Cref{fig:stock-smoothed-coordinates}.

The Gaussian working likelihood also favors BCP among the sparse
reconstructions. We compute descriptive plug-in BIC scores, counting
retained off-diagonal coordinates as parameters. All four sparse estimates
have the same parameter count and hence the same BIC penalty, so BCP's
higher working likelihood gives it the lowest BIC. The original sample
correlation fits best but has a higher BIC than BCP because of its larger
parameter count. These scores are descriptive: the sparse estimates are
not refitted by maximum likelihood, and the criterion does not adjust for
selection of the zero restrictions or serial dependence.

\section{Conclusion}

BCP is a global real-analytic parametrization of positive-definite
correlation matrices by unrestricted coordinates. Its central advantage is
that both the transformation and its inverse are available in closed form
and require no iterative scaling or diagonal completion. Balancing makes
the complete transformed matrix hollow. The rigidity theorem
shows that, for dimensions three and higher, direct hollowness and unit
local normalization uniquely determine BCP within the class of primary
spectral transformations of balanced representatives.

The coordinates combine marginal and partial correlations, respect
variable symmetries, preserve block structure, and express normalized
correlation inversion as sign reversal. Reduced block formulas and
simplified expressions for arrowhead, Gaussian-tree, moving-average,
autoregressive, and factor-model correlation matrices facilitate
evaluation for familiar dependence structures. Explicit derivatives
support estimation and inference.

A further contribution is asymptotic preservation of marginal-correlation
sparsity. Under diverging strength of the weakest factor direction and
diagonal residual covariance with variances bounded away from zero, BCP
coordinates uniformly approximate positively rescaled correlations. This yields support recovery
when the smallest nonzero signals dominate approximation and estimation
errors, a general coordinate-thresholding guarantee not shared by GFT, LS,
or InvC under the same factor assumptions. With a bounded number of factors
and heterogeneous loadings, zeroing coordinates at population correlation
zeros also yields uniformly accurate reconstruction at the same asymptotic
rate. The equity-return illustration complements these results, showing
that BCP retains correlation magnitudes and supports accurate
positive-definite reconstructions.

These results motivate further work on dynamic correlation models and
sparse estimation in BCP coordinates. The unrestricted parametrization and
explicit inverse allow correlation dynamics and regularization to be
specified while maintaining positive definiteness. Developing estimation
procedures and studying their statistical properties are natural next
steps, including data-driven threshold selection and more general
reconstruction guarantees.

\section*{Use of AI tools}

During the preparation of this work, the authors used generative AI tools to
assist with manuscript review, copy-editing, and numerical checks of
mathematical derivations. The authors reviewed and revised the AI-assisted
material and take full responsibility for the final manuscript.

\bibliographystyle{apalike2}
\bibliography{prh}

\clearpage
\newcommand{\SupplementCounter}[2]{%
  \setcounter{#1}{0}%
  \expandafter\def\csname the#1\endcsname{#2\arabic{#1}}%
  \expandafter\def\csname theH#1\endcsname{#2\arabic{#1}}}
\newcommand{\SupplementCounterS}[1]{\SupplementCounter{#1}{S}}
\setcounter{page}{1}
\renewcommand{\thepage}{S\arabic{page}}
\SupplementCounter{section}{S}
\SupplementCounter{equation}{S.}
\SupplementCounter{table}{S}
\SupplementCounter{figure}{S}
\SupplementCounter{footnote}{S}
\renewcommand{\thefootnote}{\arabic{footnote}}
\forcsvlist{\SupplementCounterS}{theorem,proposition,corollary,lemma,assumption,definition,example,remark}
\begingroup
\setstretch{1}
\null\vskip 2em
\begin{center}
{\LARGE Supplement to\\[0.5ex]
\textbf{A Balanced Parametrization of Correlation Matrices}\par}
\vskip 1.5em
{\large\lineskip .5em
\begin{tabular}[t]{c}\BCPauthor\end{tabular}\par}
\vskip 1em
{\large\BCPdate}
\end{center}
\par\vskip 1.5em
\noindent
This supplement provides proofs, extensions, and computational details
for the main paper. Notation follows \Cref{sec:background}, with additional
conventions in \Cref{app:notation}. References without an S prefix refer to the main paper's sections, results,
equations, figures, and tables; the S prefix identifies material in this
supplement.
\par
\endgroup

\section{Notation and preliminaries}\label{app:notation}

We retain the notation of \Cref{sec:background}. For square matrices,
$\diag(X_1,\ldots,X_k)$ denotes block-diagonal assembly. The Frobenius
inner product is $\ip{X}{Y}_F=\tr(X'Y)$, and adjoints are taken with
respect to this inner product.

The open elliptope $\CorP{n}$ has tangent space $\Hol{n}$ at every point.
Correlation normalization identifies matrices that differ by positive
diagonal rescaling, giving
$\CorP{n}\simeq\Symp{n}/\operatorname{Diag}^{+}(n)$, where
$\operatorname{Diag}^{+}(n)$ is the group of positive diagonal matrices
acting by $\Sigma\mapsto\Delta\Sigma\Delta$
\citep{DavidGu:2019,Thanwerdas:2024}.

\section{Proofs for the main results}\label{app:proofs}

This section proves the construction, structural, differential, and
structured-matrix results in the main paper. The proofs follow the order of
the corresponding results. The local expansions are given in
\Cref{app:expansions}.

\Needspace{5\baselineskip}
\subsection{Balancing and spectral characterization}
\label{app:balancing-characterization}

\begin{proof}[Proof of \Cref{thm:balanced-section}]
Let $\Delta=\diag(\delta_1,\ldots,\delta_n)$ and $A=\Delta C\Delta$.  Since
$C_{ii}=1$,
$A_{ii}=\delta_i^2$ and $(A^{-1})_{ii}=\delta_i^{-2}(C^{-1})_{ii}$.
Condition \eqref{eq:balanced-condition} is therefore equivalent to
$\delta_i^4=(C^{-1})_{ii}$ for $i=1,\ldots,n$.
Moreover,
\[
 \sum_{i=1}^n
 \left[\log A_{ii}-\log(A^{-1})_{ii}\right]^2
 =\sum_{i=1}^n
 \left[4\log\delta_i-\log\{(C^{-1})_{ii}\}\right]^2,
\]
so the same equations give the unique minimizer of the primal--dual
discrepancy displayed above.
The corresponding general componentwise scaling appears in the proof of
Lemma~2.4 of \citet[pp.~41--42]{Fiedler:1964}; here we use only its
correlation-matrix specialization.  Positivity selects the unique solution
$\delta_i=\{(C^{-1})_{ii}\}^{1/4}$, which proves \eqref{eq:AC}; direct
substitution gives \eqref{eq:balanced-diagonal}.  A signed permutation only
permutes the diagonal entries of $C^{-1}$, because its sign factors cancel on
the diagonal, which proves the first identity in
\eqref{eq:balanced-inversion}.

Finally,
$\cinv(C)=D^{-1/2}C^{-1}D^{-1/2}$ and $\cinv(C)^{-1}=D^{1/2}CD^{1/2}$.
The latter matrix has diagonal $D$.  Applying \eqref{eq:AC} to
$\cinv(C)$ yields
$A_{\cinv(C)}=D^{-1/4}C^{-1}D^{-1/4}=A_C^{-1}$,
and proves \eqref{eq:balanced-inversion}.
\end{proof}

\begin{proof}[Proof of \Cref{thm:rigidity-balanced-spectral}]
Let $m=n-1$ and $P=\one_n\one_n'/n$.  For $t\in\mathbb R$, define
$B_t=-mtP+t(I_n-P)=tI_n-ntP$.
The matrix $B_t$ is hollow.  Its eigenvalue on the span of $\one_n$ is
$-mt$, and its eigenvalue on the orthogonal complement is $t$, with
multiplicity $m$.  Set
$A_t=B_t+\sqrt{I_n+B_t^2}$.
Then $A_t\succ0$ and
$A_t^{-1}=\sqrt{I_n+B_t^2}-B_t$.  Consequently,
$\Diag(A_t)-\Diag(A_t^{-1})=2\Diag(B_t)=0$,
so $A_t\in\Bal{n}$.

Let $r(t)=t+\sqrt{1+t^2}$ and $f(t)=g\{r(t)\}$.  The spectral projectors of
$B_t$ give
$g(A_t)=f(-mt)P+f(t)(I_n-P)$.
Every diagonal entry of this matrix equals
$\{f(-mt)+mf(t)\}/n$.  Assumption \eqref{eq:rigidity-hollow} therefore implies
\begin{equation}\label{eq:rigidity-dilation-one}
 f(-mt)+mf(t)=0,
 \qquad t\in\mathbb R.
\end{equation}
In particular, $f(0)=0$.  Replacing $t$ by $-mt$ in
\eqref{eq:rigidity-dilation-one} and applying that identity once more gives
$f(m^2t)=m^2f(t)$.
For $t\ne0$ and every positive integer $k$,
$\frac{f(t)}{t} =\frac{f(t/m^{2k})}{t/m^{2k}}$.
Because $n\ge3$, we have $m>1$, so the argument on the right converges to
zero.  The function $r$ satisfies $r(0)=1$ and $r'(0)=1$; differentiability
of $g$ at $1$ therefore gives $f'(0)=g'(1)$.  Letting $k\to\infty$ yields
$f(t)=g'(1)t$ for every $t$, including $t=0$.

Finally, $r$ is a bijection from $\mathbb R$ onto $(0,\infty)$ with
$r^{-1}(a)=(a-a^{-1})/2$.  Hence
$g(a)=f\{r^{-1}(a)\} =\frac{g'(1)}2(a-a^{-1})$,
which proves the theorem.
\end{proof}

\begin{remark}[Scope of the characterization]
\label{rem:postcomposition-family}
For $\kappa>0$, the scalar map
$\phi_\kappa(x)=\sinh\{\kappa\arsinh(x)\}/\kappa$ is an odd analytic
diffeomorphism with $\phi_\kappa'(0)=1$. Applied entrywise to $B(C)$,
it preserves signed-permutation equivariance, inverse consistency, and
block patterns, and has an explicit inverse. For $\kappa\ne1$ it lies
outside the primary spectral class in \Cref{thm:rigidity-balanced-spectral}.
Thus that theorem does not assert uniqueness among all parametrizations
with these operational properties.

Within the affine class $\varphi(A)=\alpha A+\beta A^{-1}+\gamma I_n$,
the identity $\varphi(A^{-1})=-\varphi(A)$ for every SPD matrix $A$ forces
$\beta=-\alpha$ and $\gamma=0$. Unit local normalization then gives
$\alpha-\beta=1$, so $\varphi(A)=(A-A^{-1})/2$. The rigidity theorem does not
assume this affine form.
\end{remark}

\subsection{Proof of the global diffeomorphism}
\label{app:global-diffeo-proof}

\begin{proof}[Proof of \Cref{thm:global-diffeo}]
For $C\in\CorP{n}$, its balanced representative $A=A_C$ satisfies
$\Cor(A)=C$ and $\Diag(A)=\Diag(A^{-1})$. Hence
$B=(A-A^{-1})/2$ is symmetric and hollow, and
\[
 I_n+B^2=\tfrac14(A+A^{-1})^2,
 \qquad B+\sqrt{I_n+B^2}=A,
\]
because $A+A^{-1}\succ0$. Correlation normalization gives $C\{B(C)\}=C$.

Conversely, for $B\in\Hol{n}$, set $H=\sqrt{I_n+B^2}$ and $A=H+B$.
Spectral calculus gives $A\succ0$ and $A^{-1}=H-B$, so
$\Diag(A)=\Diag(A^{-1})=:E$. For $C=E^{-1/2}AE^{-1/2}$,
\[
 C^{-1}=E^{1/2}A^{-1}E^{1/2},\qquad
 D=E^2,\qquad A_C=D^{1/4}CD^{1/4}=A.
\]
Thus $B\{C(B)\}=(A-A^{-1})/2=B$. The same identities establish the
balanced-to-hollow bijection.

Inversion, positive diagonal powers, the principal square root on the
positive-definite cone, and correlation normalization are real analytic
\citep{Higham:2008}. Both maps are compositions of these operations on
the stated domains, completing the proof.
\end{proof}
\subsection{Structural properties}

\begin{proof}[Proof of \Cref{thm:signed-equivariance}]
Since $(QCQ')^{-1}=QC^{-1}Q'$, the diagonal of the inverse is permuted but unaffected by sign changes.  Consequently
$\Diag\{(QCQ')^{-1}\}^{1/4}=QD^{1/4}Q'$,
where the equality is understood as equality of diagonal matrices after the sign factors cancel.  It follows that
$A_{QCQ'}=QA_CQ'$.
The first identity in \eqref{eq:BFT-equivariance} follows immediately.

For the inverse, orthogonal functional calculus gives
$\sqrt{I_n+(QBQ')^2}=Q\sqrt{I_n+B^2}Q'$.
Correlation normalization commutes with signed permutations, yielding the second identity in \eqref{eq:BFT-equivariance}.
\end{proof}

\begin{proof}[Proof of \Cref{thm:inverse-consistency}]
Let $K=C^{-1}$ and $D=\Diag(K)$.  Then
$\cinv(C)=D^{-1/2}KD^{-1/2}$ and $\cinv(C)^{-1}=D^{1/2}CD^{1/2}$.
Because $\Diag(C)=I_n$, we have $\Diag(\cinv(C)^{-1})=D$.  Therefore the balanced representative of $\cinv(C)$ is
$A_{\cinv(C)} =D^{1/4}\cinv(C)D^{1/4} =D^{-1/4}KD^{-1/4} =A_C^{-1}$.
Applying $h(A)=(A-A^{-1})/2$ proves the first identity in
\eqref{eq:inverse-consistency}; the second follows from bijectivity.
\end{proof}

InvC does not share this property. For
$C=\tfrac12I_3+\tfrac12\one_3\one_3'$, the off-diagonal entries of
$C^{-1}$ are $-1/2$, whereas those of $\cinv(C)^{-1}$ are $3/4$, not $1/2$.

\begin{proof}[Proof of \Cref{prop:partial-correlation}]
Since
$\cinv(C)=D^{-1/2}C^{-1}D^{-1/2}$,
we may rewrite \eqref{eq:BFT-main-definition} as
$B(C) =\frac12D^{1/4}\{C-\cinv(C)\}D^{1/4}$.
The off-diagonal entries of $\cinv(C)$ are $-\pi_{ij}$ and
$\sqrt{\tau_i\tau_j}=(K_{ii}K_{jj})^{-1/4}$, which gives
\eqref{eq:B-R2}.
\end{proof}

\begin{proof}[Proof of \Cref{prop:blocks}]
All operations defining $D$, $A_C$, $h(A)$, the principal square root, and $\Cor$ preserve block diagonal structure.  The fixed-point statement follows from \Cref{thm:signed-equivariance} and bijectivity.
\end{proof}

\begin{proof}[Proof of \Cref{prop:kronecker}]
Write $D_k=\Diag(C_k^{-1})$ and $A_k=A_{C_k}$.  Since
$(C_1\otimes C_2)^{-1}=C_1^{-1}\otimes C_2^{-1}$,
$\Diag\{(C_1\otimes C_2)^{-1}\}=D_1\otimes D_2$.
Positive diagonal powers commute with Kronecker products, and therefore
\[
 A_{C_1\otimes C_2}
 =(D_1^{1/4}C_1D_1^{1/4})
  \otimes(D_2^{1/4}C_2D_2^{1/4})
 =A_1\otimes A_2,
\]
which proves the first identity in \eqref{eq:BFT-kronecker}.  Moreover,
\[
 S_k=\frac12(A_k-A_k^{-1}),
 \qquad
 H_k=\frac12(A_k+A_k^{-1})=\sqrt{I_{n_k}+S_k^2}.
\]
Hence $A_k=H_k+S_k$ and $A_k^{-1}=H_k-S_k$.  Expanding
$\frac12\{A_1\otimes A_2-A_1^{-1}\otimes A_2^{-1}\}$
gives $S_1\otimes H_2+H_1\otimes S_2$, proving the second identity in
\eqref{eq:BFT-kronecker}.
\end{proof}

\begin{proof}[Proof of \Cref{prop:properness}]
Let $\mathcal H\subset\Hol{n}$ be compact.  Since $B\mapsto C(B)$ is continuous, $C(\mathcal H)=\{C(B):B\in \mathcal H\}$ is compact and is contained in $\CorP{n}$.  Because $C\mapsto B(C)$ and $B\mapsto C(B)$ are inverse homeomorphisms, $C(\mathcal H)$ is exactly the inverse image of $\mathcal H$ under $C\mapsto B(C)$, proving properness.

For \eqref{eq:properness-bound}, let $A=A_C$ have eigenvalues $a_1,\ldots,a_n$ and set $E=\Diag(A)$.  Since $B=\tfrac12(A-A^{-1})$ is a primary matrix function of $A$, its eigenvalues are $\sinh(\log a_i)$, so
\[
 \norm{B}_2
 =\sinh\left(\log\kappa\right),
 \qquad
 \kappa:=\max_i\max\left(a_i,a_i^{-1}\right).
\]
For the lower bound, every diagonal entry of $A$ lies between $\lambda_{\min}(A)$ and $\lambda_{\max}(A)$, so for a unit vector $x$, writing $y=E^{-1/2}x$,
\[
 x'Cx
 =y'Ay
 \geq\lambda_{\min}(A)\norm y^2
 \geq\frac{\lambda_{\min}(A)}{\lambda_{\max}(A)}.
\]
Hence $\lambda_{\max}(A)/\lambda_{\min}(A)\geq1/\ell$, and since $\kappa^2\geq\lambda_{\max}(A)/\lambda_{\min}(A)$ we get $\kappa\geq\ell^{-1/2}$; monotonicity of $\sinh$ gives the left inequality, and solving it for $\ell$ gives \eqref{eq:lambda-min-bound}.
For the upper bound, $(C^{-1})_{ii}=(1-R_i^2)^{-1}\geq1$ gives $D\succeq I_n$, so $x'Ax=(D^{1/4}x)'C(D^{1/4}x)\geq\ell\norm x^2$ and $\lambda_{\min}(A)\geq\ell$.  Also $\lambda_{\max}(A)\leq\lambda_{\max}(C)\max_i(C^{-1})_{ii}^{1/2}\leq n\ell^{-1/2}$, using $\lambda_{\max}(C)\leq\tr C=n$ and $(C^{-1})_{ii}\leq1/\ell$.  Because $\ell\leq1$, both $n\ell^{-1/2}$ and $1/\ell$ are at most $n/\ell$, so $\kappa\leq n/\ell$ and the right inequality follows.
Finally, if \eqref{eq:boundary-divergence} failed, a bounded subsequence of
transformed matrices would have a convergent subsequence.  Its image under
$B\mapsto C(B)$ would converge to a positive-definite correlation matrix,
contradicting $\lambda_{\min}(C_k)\to0$.  The result also follows immediately
from the left inequality in \eqref{eq:properness-bound}.
\end{proof}

\subsection{Differentials and inference}

\begin{proof}[Proof of \Cref{prop:differential-forward}]
The inverse differential is $\dd K[X]=-KXK$, so
$\dd D[X]=-\Diag(KXK)$.
Because $D$ and $\Delta$ are diagonal,
$\Delta^{-1}\dd\Delta[X] =\frac14D^{-1}\dd D[X] =U_X$.
Differentiating $A=\Delta C\Delta$ therefore gives \eqref{eq:ZX-forward}.  Finally,
$\dd(A^{-1})[Z]=-A^{-1}ZA^{-1}$, which gives \eqref{eq:dPsiBFT-compact}.  At $C=I_n$, $K=D=\Delta=A=I_n$ and $\diag(X)=0$, hence $U_X=0$ and the result follows.
\end{proof}

\begin{proof}[Proof of \Cref{prop:differential-inverse}]
The derivative of a symmetric matrix function acts on the eigenbasis by the Hadamard product with the matrix of first divided differences \citep[Theorem V.3.3]{Bhatia:1997}, so it suffices to identify that matrix.  Let $h(a)=(a-a^{-1})/2$, so $B=h(A)$ and $A=f(B)$ with $f=h^{-1}$.  For $a_i,a_j>0$,
$h^{[1]}(a_i,a_j) =\frac12\left(1+\frac1{a_ia_j}\right)$.
The divided difference of the inverse function is therefore
\[
 f^{[1]}(\lambda_i,\lambda_j)
 =\{h^{[1]}(a_i,a_j)\}^{-1}
 =\frac{2a_ia_j}{1+a_ia_j},
\]
including the repeated-eigenvalue case by continuity.  This proves \eqref{eq:ZY-inverse}.  Differentiating $C=RAR$ with $R=\diag\{\diag(A)^{-1/2}\}$ gives \eqref{eq:dPsiBFTinv-explicit}.  Since $a_ia_j>0$, the bound $0<\Omega_{ij}<2$ is immediate.
\end{proof}

\begin{proof}[Proof of \Cref{prop:jacobian-determinant}]
Let $A=B+\sqrt{I_n+B^2}$ and
$e=\diag(A)$. The derivative of this spectral lift on $\Sym{n}$ is
\[
 \mathcal L(Y)=Q\{\Omega\Had(Q'YQ)\}Q',
 \qquad
 \mathcal L^{-1}(Z)=\tfrac12(Z+A^{-1}ZA^{-1}).
\]
In an eigenbasis of $A$, the diagonal and symmetric off-diagonal
basis elements have multipliers $\Omega_{ii}$ and $\Omega_{ij}$,
respectively. Hence
$\det\mathcal L=\prod_{i\leq j}\Omega_{ij}>0$.

Define
\[
 \mathcal K_A=\{DA+AD:D\text{ diagonal}\},\qquad
 \mathcal M(Z)=Z-\tfrac12(V_ZA+AV_Z),\quad
 V_Z=\diag\{\diag(Z)\oslash e\}.
\]
Since $\diag(DA+AD)=2De$ and every $e_i>0$,
$\Sym{n}=\Hol{n}\oplus\mathcal K_A$ and $\mathcal M$ is the projection
onto $\Hol{n}$ along $\mathcal K_A$. Let $\mathcal N(Z)=Z-\mathcal M(Z)$
be the complementary projection. For $Z=DA+AD$, balancing gives
\[
 \diag\{\mathcal L^{-1}(Z)\}
 =D\{e+\diag(A^{-1})\}=2De=\diag(Z).
\]
The projection $\mathcal N$ is determined solely by the diagonal of
its argument, so $\mathcal N\mathcal L^{-1}(Z)=Z$ on $\mathcal K_A$.
Thus the $\mathcal K_A$ diagonal block of $\mathcal L^{-1}$, in this
direct-sum decomposition, is the identity. The complementary-minor
identity for an invertible block matrix therefore yields
\[
 \det(\mathcal M\mathcal L|_{\Hol{n}})
 =\det\mathcal L\,
   \det(\mathcal N\mathcal L^{-1}|_{\mathcal K_A})
 =\det\mathcal L.
\]
By \Cref{prop:differential-inverse}, the inverse derivative is
$Y\mapsto R\mathcal M\mathcal L(Y)R$, where
$R=\diag(e^{-1/2})$. On hollow matrices, the final scaling multiplies
coordinate $(i,j)$ by $(e_ie_j)^{-1/2}$. Its determinant is
\[
 \prod_{i<j}(e_ie_j)^{-1/2}
 =\prod_i e_i^{-(n-1)/2}.
\]
Combining the two determinants proves \eqref{eq:jacobian-determinant}
in $\vecl$ coordinates. The forward determinant follows from the chain
rule for the inverse maps.
\end{proof}

\begin{proof}[Proof of \Cref{thm:delta}]
The transformation is differentiable by \Cref{thm:global-diffeo}.  Applying
the multivariate delta method to $C\mapsto\vecl\{B(C)\}$ gives
\eqref{eq:B-CLT}.  Equation \eqref{eq:derivative-identity} gives $J_{I_n}=I_d$,
which proves the result.
\end{proof}

\subsection{Block-canonical preliminaries}
\label{app:block-preliminaries}

The fixed-point characterization and exact-preservation proof below apply
to every partition with $n_k\ge1$. A singleton block has no off-diagonal
entry or contrast space.

Let $\mathfrak G_{\mathcal P}$ be the group of block-diagonal permutation
matrices $\diag(P_1,\ldots,P_G)$, where $P_k$ independently permutes the
$n_k$ variables in block $k$.

\begin{lemma}[Block symmetry as a fixed-point condition]
\label{lem:block-fixed-point}
For a symmetric matrix $X$, the following are equivalent:
\begin{enumerate}[label=(\roman*),leftmargin=2em]
 \item $PXP'=X$ for every $P\in\mathfrak G_{\mathcal P}$;
 \item every diagonal block has a common diagonal and, when $n_k\ge2$,
 a common off-diagonal entry, while every off-diagonal block has constant entries.
\end{enumerate}
Consequently, $\CorBlock{n}$ and $\HolBlock{n}$ are the fixed-point sets of
$\mathfrak G_{\mathcal P}$ in $\CorP{n}$ and $\Hol{n}$, respectively.
\end{lemma}

For the reduced calculations, the following canonical formulas are written
for $n_k\ge2$; the singleton modifications are given in
\Cref{app:block-singletons}. Set
\[
 v_k=\frac{1}{\sqrt{n_k}}\one_{n_k},
 \qquad
 P_{[k,l]}=v_kv_l',
 \qquad
 P_{[k,k]}^\perp=I_{n_k}-P_{[k,k]}.
\]
Choose an orthogonal matrix $Q_{\mathcal P}$ whose first $G$ columns are the
embedded vectors $v_1,\ldots,v_G$ and whose remaining columns are orthonormal
within-block contrasts.  The canonical representation of
\citet{ArchakovHansen:CanonicalBlockMatrix} writes a block-symmetric matrix as
\begin{equation*}
 \begin{aligned}
 X_{[k,l]}&=a_{kl}P_{[k,l]},&&k\ne l,\\
 X_{[k,k]}&=a_{kk}P_{[k,k]}+\lambda_kP_{[k,k]}^\perp,
 &&k=1,\ldots,G,
 \end{aligned}
\end{equation*}
and hence
\begin{equation*}
 Q_{\mathcal P}'XQ_{\mathcal P}
 =R\oplus\bigoplus_{k=1}^{G}
 \lambda_kI_{n_k-1},
 \qquad R=(a_{kl})_{k,l=1}^G.
\end{equation*}
Whenever a matrix function $f$ is defined,
\begin{equation}\label{eq:AH-functional-calculus}
 f(X)=Q_{\mathcal P}
 \left\{f(R)\oplus\bigoplus_{k=1}^{G}
 f(\lambda_k)I_{n_k-1}\right\}Q_{\mathcal P}'.
\end{equation}
Thus inversion, powers, logarithms, exponentials, and the principal square
root preserve the block pattern.  These general facts are used only to derive
the BCP-specific formulas in \Cref{thm:block-forward,thm:block-inverse}.

\begin{proof}[Proof of \Cref{lem:block-fixed-point}]
If $X$ has the stated block form, independent within-block permutations
leave it unchanged. Conversely, these permutations act transitively on the
diagonal positions of each block and, when $n_k\ge2$, on its off-diagonal
positions. Invariance therefore makes the entries constant within each
of these sets. Independent permutations in blocks $k$ and $l$ act
transitively on the positions of $X_{[k,l]}$, so that block is constant
as well. This includes singleton blocks, for which the corresponding
permutation group is trivial. Intersecting this fixed-point space with
$\CorP{n}$ or $\Hol{n}$ gives the stated block classes.
\end{proof}

\begin{proof}[Proof of \Cref{thm:block-restriction}]
By \Cref{thm:signed-equivariance}, both $C\mapsto B(C)$ and $B\mapsto C(B)$ commute with
congruence by every $P\in\mathfrak G_{\mathcal P}$.  Hence each map sends the
fixed-point set of $\mathfrak G_{\mathcal P}$ into the corresponding fixed-point
set.  The characterization in \Cref{lem:block-fixed-point} gives the two block
classes.  Since the global transformation and its inverse are mutually inverse
analytic diffeomorphisms,
their restrictions are also mutually inverse analytic diffeomorphisms.
The argument imposes no lower bound on block size beyond $n_k\ge1$.
\end{proof}

\begin{proof}[Proof of \Cref{thm:block-forward}]
The canonical representation gives
\[
 Q_{\mathcal P}'C^{-1}Q_{\mathcal P}
 =R^{-1}\oplus\bigoplus_{k=1}^{G}
 \lambda_k^{-1}I_{n_k-1}.
\]
The diagonal of
$(R^{-1})_{kk}P_{[k,k]}+\lambda_k^{-1}P_{[k,k]}^\perp$
is constant and equal to $\delta_k^4$ in \eqref{eq:block-dk}.  This proves the
asserted blockwise form of $D$.

Because $D^{1/4}$ is constant within each block, its canonical action is
\[
 Q_{\mathcal P}'D^{1/4}Q_{\mathcal P}
 =\Delta\oplus\bigoplus_{k=1}^{G}
 \delta_kI_{n_k-1}.
\]
It follows that the balanced representative $A_C=D^{1/4}CD^{1/4}$ has
canonical form
\[
 \Delta R\Delta\oplus\bigoplus_{k=1}^{G}
 \delta_k^2\lambda_kI_{n_k-1},
\]
and its inverse has the corresponding form with $\Delta$ and $\delta_k$
inverted and $R$ and $\lambda_k$ replaced by their inverses.  Taking half
their difference gives the two quantities in
\eqref{eq:block-forward-components}.

The diagonal entries of the $k$th diagonal block equal
$\frac{M_{kk}}{n_k} +\left(1-\frac1{n_k}\right)\zeta_k$.
They vanish because $B$ is hollow.  Reconstructing the blocks from the
projectors gives \eqref{eq:block-eta-forward}.
\end{proof}

\begin{proof}[Proof of \Cref{thm:block-inverse}]
By the functional-calculus identity \eqref{eq:AH-functional-calculus},
\[
 B+\sqrt{I_n+B^2}
 =Q_{\mathcal P}
 \left\{F\oplus\bigoplus_{k=1}^{G}
 a_kI_{n_k-1}\right\}
 Q_{\mathcal P}'.
\]
Its diagonal entries in block $k$ equal $e_k$ in \eqref{eq:block-ek}.
Correlation normalization multiplies block $(k,l)$ by
$e_k^{-1/2}e_l^{-1/2}$.  Expanding the projectors in the canonical
representation yields \eqref{eq:block-inverse-correlations}.  The diagonal
entries are one by construction.
\end{proof}

\begin{proof}[Proof of \Cref{cor:block-complexity}]
The forward formulas in \Cref{thm:block-forward} require the inverse of the
$G\times G$ core $R$ and then only diagonal and scalar operations.  The inverse
formulas in \Cref{thm:block-inverse} require one symmetric $G\times G$
eigendecomposition, used to evaluate $F=e^{\arsinh H}$ on the core $H$ of
\eqref{eq:block-H-zeta-input}, and then only scalar operations.  The unrestricted block
parameters consist of one value for each unordered pair of blocks and one
within-block parameter for every nonsingleton block, giving dimension
$G(G-1)/2+|\{k:n_k>1\}|$. For singleton blocks, the rules in
\Cref{app:block-singletons} omit the absent contrast terms and retain the
same $G\times G$ core calculations and computational cost.
\end{proof}

\subsection{Equicorrelation formulas}\label{app:equicorrelation}

\begin{proof}[Proof of Proposition~\ref{prop:equicorrelation} and
identity~\eqref{eq:AH-BFT-equicorr-link}]
For
$C=(1-\rho)I_n+\rho\one\one'$,
the inverse is
$C^{-1} =\frac1{1-\rho}I_n -\frac{\rho}{(1-\rho)\{1+(n-1)\rho\}}\one\one'$.
Hence
$D=\delta^4I_n$ with $\delta^4=\frac{1+(n-2)\rho}{(1-\rho)\{1+(n-1)\rho\}}$.
Both $A_C=\delta^2C$ and $A_C^{-1}=\delta^{-2}C^{-1}$ have constant diagonal and constant off-diagonal entries.  Their off-diagonal half-difference is
\[
 b
 =\frac12\left[
 \rho\delta^2
 +\frac{\rho}{(1-\rho)\{1+(n-1)\rho\}\delta^2}
 \right],
\]
which simplifies to \eqref{eq:BFT-equicorr}.

Now $B=b(\one\one'-I_n)$ has eigenvalue $(n-1)b$ in the span of $\one$ and eigenvalue $-b$ on its orthogonal complement.  Therefore
$A=e^{\arsinh B}$
has corresponding eigenvalues
$a_1=e^{\arsinh((n-1)b)}$ and $a_2=e^{-\arsinh(b)}$.
Since $C=\Cor(A)$ and $A$ has a constant diagonal, the ratio of the two distinct eigenvalues of $C$ equals $a_1/a_2$.  But the same ratio is
$\frac{1+(n-1)\rho}{1-\rho}$.
Taking logarithms gives
$\log\frac{1+(n-1)\rho}{1-\rho} =\arsinh((n-1)b)+\arsinh(b)$,
which proves \eqref{eq:AH-BFT-equicorr-link}.

Finally, a matrix with eigenvalues $a_1$ and $a_2$ in the two invariant subspaces has correlation coefficient
$\rho=\frac{a_1-a_2}{a_1+(n-1)a_2}$,
which gives \eqref{eq:rho-from-b}.  Since the global map is a diffeomorphism and preserves the one-dimensional equicorrelation class, $b_n$ is a smooth bijection onto $\R$; monotonicity also follows directly from the displayed inverse.
\end{proof}
\subsection{BCP for selected structured correlation matrices}

\begin{proof}[Proof of \Cref{prop:star-coordinate}]
The Schur complement of \eqref{eq:star-correlation-matrix} gives
\[
 (C^\star)^{-1}=\omega^{-2}
 \begin{pmatrix}1&-\rho'\\-\rho&\omega^2I_r+\rho\rho'\end{pmatrix},
 \qquad D_\star=\omega^{-2}\diag(1,R_\rho^4).
\]
Substitution into \eqref{eq:BFT-main-definition} gives the factor-to-leaf
entries $\eta_\rho$ in \eqref{eq:BFT-star-block-form}. For distinct
leaves $j,k$, $C^\star_{j+1,k+1}=0$ and the precision contribution is
$-\rho_j\rho_k/(2\omega h_j^{1/4}h_k^{1/4})
=-v_{\rho,j}v_{\rho,k}/(2\omega)$.
The diagonal is zero by balancing, proving the lower block and
\eqref{eq:BFT-star-factor-elements}.
\end{proof}

\begin{proof}[Proof of \Cref{prop:tree-correlation}]
Choose an arbitrary root and orient every edge away from it.  A standardized
Gaussian tree can be generated recursively as
$X_j=\theta_{ij}X_i+\sqrt{1-\theta_{ij}^2}\,\varepsilon_j$
when $i$ is the parent of $j$, with independent standard-normal innovations.
Iterated conditioning gives the path-product correlation in
\eqref{eq:tree-correlation}. Expanding the negative Gaussian log density,
up to an additive constant,
\[
 \frac12X_{\mathrm{root}}^2
 +\frac12\sum_{j\ne\mathrm{root}}
 \frac{(X_j-\theta_{p(j)j}X_{p(j)})^2}
 {1-\theta_{p(j)j}^2}
\]
gives the off-diagonal and diagonal precision entries in
\eqref{eq:tree-precision}, independently of the chosen root.
Substituting these precision entries and the path-product correlations
into \eqref{eq:BFT-main-definition} proves
\eqref{eq:BFT-tree-correlation}.
\end{proof}

\begin{proof}[Proof of \Cref{prop:ma-correlation}]
Set $\theta_k=0$ for $k<0$ or $k>q$ and let
$L=(\theta_{i-j})_{i,j=1}^n$. With
$\varepsilon=(\varepsilon_1,\ldots,\varepsilon_n)'$ and
$\varepsilon_-=(\varepsilon_0,\varepsilon_{-1},\ldots,
\varepsilon_{1-q})'$, the observation vector satisfies
\[
 x=L\varepsilon+H\varepsilon_-.
\]
The innovation vectors have identity covariance matrices and are
uncorrelated with one another. Since each $x_t$ has variance $s$,
\[
 C=s^{-1}(LL'+HH').
\]
The matrix $L$ is unit lower triangular, so this representation proves
positive definiteness for every finite $n$, irrespective of the roots
of the MA polynomial. The convolution identity
\[
 \sum_{\ell=0}^{\min(q,k)}\theta_\ell\psi_{k-\ell}
 =\indic\{k=0\}
\]
implied by \eqref{eq:MAq-filter-recursion} gives $F=L^{-1}$.
Applying the Woodbury identity yields
\begin{align*}
 C^{-1}
 &=sF'\left(I_n+FHH'F'\right)^{-1}F\\
 &=s\left\{F'F-F'FH
       \left(I_q+H'F'FH\right)^{-1}H'F'F\right\}\\
 &=s\left(W-WHM^{-1}H'W\right).
\end{align*}
Here $M$ is positive definite because $W=F'F$ is positive definite.
Substituting this precision matrix and $C_{ij}=\rho_{|i-j|}$ into
\eqref{eq:BFT-main-definition} gives \eqref{eq:BFT-MAq}; its specialization
outside the correlation band follows from $\rho_h=0$ for $h>q$.
\end{proof}

\begin{proof}[Proof of \Cref{prop:ar-correlation}]
Let $x=(x_1,\ldots,x_n)'$ and let $T$ be the unit lower-triangular
matrix with first $p$ rows $E_p$ and remaining rows
$r_{p+1}',\ldots,r_n'$. Then
$Tx=(x_1,\ldots,x_p,\varepsilon_{p+1},\ldots,\varepsilon_n)'$.
Causality and the white-noise assumption imply that the first $p$
components are uncorrelated with the remaining innovations, which
are mutually uncorrelated. Hence
\[
 TCT'=\begin{pmatrix}C_p&0\\0&\sigma^2I_{n-p}\end{pmatrix},
 \qquad
 K=T'\begin{pmatrix}C_p^{-1}&0\\0&\sigma^{-2}I_{n-p}\end{pmatrix}T.
\]
Expanding the second identity gives \eqref{eq:ARp-precision}, without
any distributional assumption beyond the stated second moments.
The initial-block term is supported on the first $p$ coordinates,
and each $r_s$ is supported on $\{s-p,\ldots,s\}$.
Hence $K_{ij}=0$ for $|i-j|>p$.
Substitution in \eqref{eq:BFT-main-definition} proves
\eqref{eq:BFT-ARp}.

\Needspace{3\baselineskip}
For completeness, the precision entries have simple scalar expressions
away from the boundaries. Put
\[
 a=1+\sum_{\ell=1}^p\phi_\ell^2,
 \qquad
 g_h=
 \begin{cases}
 -\phi_h+\displaystyle\sum_{\ell=1}^{p-h}
                 \phi_\ell\phi_{\ell+h},&1\le h\le p,\\[4pt]
 0,&h>p.
 \end{cases}
\]
For $p+1\le i<j\le n-p$, all terms contributing to the interior
precision entries are present in the sum in \eqref{eq:ARp-precision},
while the initial-block term is zero. Collecting their coefficients
gives $K_{ii}=K_{jj}=a/\sigma^2$ and $K_{ij}=g_{j-i}/\sigma^2$.
Consequently,
\begin{equation}\label{eq:ARp-interior-BFT}
 [B(C)]_{ij}
 =\frac{a\rho_h-g_h}{2\sqrt{a\sigma^2}},
 \qquad h=j-i,\quad p+1\le i<j\le n-p.
\end{equation}
Thus the interior entries depend only on the lag, although the full
transformed matrix need not be Toeplitz.

For $p=1$, $C_p=(1)$, $\phi_1=\rho$, and $\sigma^2=1-\rho^2$.
The precision formula gives the endpoint and interior diagonal entries
stated in \Cref{subsec:ar-correlations}, the adjacent entries
$K_{i,i+1}=-\rho/(1-\rho^2)$, and zeros elsewhere off the diagonal.
These entries give \eqref{eq:BFT-AR1}, equivalently obtained from
\Cref{prop:tree-correlation} on a homogeneous path.
\end{proof}

\begin{proof}[Proof of \Cref{prop:factor-correlation}]
The Woodbury identity applied to \eqref{eq:factor-correlation} gives
$C^{-1} =Q^{-1}-Q^{-1}\Lambda M^{-1}\Lambda'Q^{-1}$.
Its $i$th diagonal entry is exactly $\delta_i^4$ in
\eqref{eq:factor-M-delta}.  For $i\ne j$,
\[
 C_{ij}=\lambda_i'\lambda_j,
 \qquad
 (C^{-1})_{ij}
 =-\frac{\lambda_i'M^{-1}\lambda_j}{q_iq_j}.
\]
Substitution into \eqref{eq:BFT-main-definition} proves
\eqref{eq:BFT-factor-correlation}.  The definitions in
\eqref{eq:factor-dual-loadings} give
\[
 u_i'u_j=\delta_i\delta_j\lambda_i'\lambda_j,
 \qquad
 v_i'v_j
 =\frac{\lambda_i'M^{-1}\lambda_j}
 {q_iq_j\delta_i\delta_j},
\]
which proves \eqref{eq:BFT-factor-Gram}.  When $r=1$, $M=\kappa$ and direct
substitution yields the one-factor formula.
\end{proof}

\section{Adjoint derivative and likelihood scores}
\label{app:adjoint-geometry}

The inverse derivative in \Cref{prop:differential-inverse} gives objective
gradients without forming a full Jacobian. For symmetric $G_C$, set
\[
 \tilde G=RG_CR,\qquad
 H_G=\tilde G-\diag\{\diag(A\tilde G)\oslash e\}.
\]
Its Frobenius adjoint is
\begin{equation}\label{eq:VJP-BFT}
 (\dd_B C)^*[G_C]
 =\Off\left[Q\{\Omega\Had(Q'H_GQ)\}Q'\right].
\end{equation}
If $\dd\ell=\tr(G_C\dd C)$ and $s=\vecl(B)$, then
$\partial\ell/\partial s=2\vecl\{(\dd_B C)^*[G_C]\}$; the factor two
accounts for the two symmetric entries associated with each coordinate.

\begin{proof}
For a symmetric direction $Z$, put $V=\diag\{\diag(Z)\oslash e\}$.
The normalization differential is
$\dd_A\Cor[Z]=R\{Z-(VA+AV)/2\}R$. Therefore
\[
 \ip{G_C}{\dd_A\Cor[Z]}_F
 =\ip{\tilde G}{Z}_F
 -\sum_i\frac{(A\tilde G)_{ii}}{e_i}Z_{ii}
 =\ip{H_G}{Z}_F.
\]
The spectral derivative $Z\mapsto Q\{\Omega\Had(Q'ZQ)\}Q'$ is
self-adjoint. Composing these adjoints and projecting onto the hollow
input space gives \eqref{eq:VJP-BFT}.
\end{proof}

\section{Relations to the generalized Fisher transformation}
\label{app:AH-relations}

This section complements \Cref{sec:AH} with the exact global transition,
a commuting-class identity, and the proof of the local comparison.

\subsection{Exact transition and diagonal commutation}

Writing $A=B+\sqrt{I_n+B^2}=e^{\arsinh B}$ and $E=\Diag(A)$, the
transition from BCP to GFT is the analytic diffeomorphism
\[
 \mathcal T(B)=\Off\log(E^{-1/2}AE^{-1/2}).
\]
A diagonal congruence generally does not commute with the matrix
logarithm. The following condition gives an exact simplification.

\begin{proposition}[Exact link under diagonal commutation]
\label{prop:commuting-link}
If $D=\Diag(C^{-1})$ commutes with $C$, then
\begin{equation}\label{eq:commuting-link}
 \Off(\log C)=\Off\{\arsinh B(C)\}.
\end{equation}
This condition is equivalent to $C_{ij}=0$ whenever
$(C^{-1})_{ii}\ne(C^{-1})_{jj}$, so the class consists of direct sums
of blocks with scalar precision diagonals. It includes equicorrelation
and symmetric circulant correlation matrices.
\end{proposition}

\begin{proof}
The condition is equivalent to $E=\Diag(A_C)=D^{1/2}$ commuting with
$A_C$. Then $C=E^{-1}A_C$ and
$\log C=\log A_C-\log E=\arsinh B(C)-\log E$.
Taking off-diagonal entries proves \eqref{eq:commuting-link}.
\end{proof}

\subsection{Local expansions}\label{app:expansions}

This section proves \Cref{thm:local-comparison}.  Throughout, the dimension
$n$ is fixed, so constants in the remainder bounds may depend on $n$.  Let
$C=I_n+E$ with $E\in\Hol{n}$,
and use a fixed submultiplicative matrix norm, as in the theorem.
Write
$G_2=\Diag(E^2)$ and $G_3=\Diag(E^3)$.

The inverse expansion is
$C^{-1}=I_n-E+E^2-E^3+\cO(\norm E^4)$.
Therefore
$D=I_n+G_2-G_3+\cO(\norm E^4)$.
Because $G_2$ and $G_3$ are respectively second and third order,
\begin{align*}
 D^{1/4}
 &=I_n+\frac14G_2-\frac14G_3+\cO(\norm E^4),
 \\
 D^{-1/4}
 &=I_n-\frac14G_2+\frac14G_3+\cO(\norm E^4).
\end{align*}
Multiplying out gives
\begin{align*}
 A_C
 &=I_n+E+\frac12G_2-\frac12G_3
   +\frac14(G_2E+EG_2)+\cO(\norm E^4),
 \\
 A_C^{-1}
 &=I_n-E+E^2-E^3-\frac12G_2+\frac12G_3
   +\frac14(G_2E+EG_2)+\cO(\norm E^4).
\end{align*}
Taking half the difference,
\begin{align*}
 B(I_n+E)
 &=E-\frac12E^2+\frac12E^3+\frac12G_2-\frac12G_3
   +\cO(\norm E^4)\\
 &=E-\frac12\Off(E^2)+\frac12\Off(E^3)
   +\cO(\norm E^4).
\end{align*}

The matrix logarithm has expansion
$\log(I_n+E)=E-\frac12E^2+\frac13E^3+\cO(\norm E^4)$.
Applying $\Off$ gives
\begin{equation*}
 \Off\{\log(I_n+E)\}=E-\frac12\Off(E^2)+\frac13\Off(E^3)
 +\cO(\norm E^4).
\end{equation*}
Subtracting proves \eqref{eq:BFT-OL-difference} and
\Cref{thm:local-comparison}.

Finally, let $B=B(I_n+E)$ and $H=\Off\{\log(I_n+E)\}$.  Since
$E=B+\cO(\norm B^2)$, \eqref{eq:BFT-OL-difference} implies
\begin{equation*}
 \mathcal T(B)=H=B-\frac16\Off(B^3)+\cO(\norm B^4).
\end{equation*}
For $E_T=\hat C_T-I_n=\cO_p(T^{-1/2})$, the cubic difference is
$\cO_p(T^{-3/2})$ and the fourth-order remainder is $\cO_p(T^{-2})$.
This gives the local asymptotic equivalence stated in \Cref{sec:AH}.
The maps are analytic, so differentiating their expansions also shows
that their Jacobians differ at order $\norm E^2$.

\section{Block formulas with singleton groups}
\label{app:block-singletons}

For a partition with singleton groups, the formulas in
\Cref{thm:block-forward,thm:block-inverse} apply to every nonsingleton
block without change. A singleton has no within-block correlation or
contrast parameter. For each $k$ with $n_k=1$, use
\[
 R_{kk}=1,\qquad \delta_k^4=(R^{-1})_{kk},\qquad
 H_{kk}=0,\qquad e_k=F_{kk}.
\]
Omit $\lambda_k$, $\zeta_k$, $a_k$, and the formula for $\rho_{kk}$
for that group. The between-block formulas remain
\[
 \eta_{kl}=\frac{M_{kl}}{\sqrt{n_kn_l}},\qquad
 \rho_{kl}=\frac{F_{kl}}{\sqrt{n_kn_le_ke_l}},\quad k\ne l.
\]
These rules follow from \eqref{eq:AH-functional-calculus} by omitting
the absent contrast spaces. The number of free block coordinates is
$G(G-1)/2+|\{k:n_k>1\}|$, and only $G\times G$ matrix functions are needed.

\section{Computational implementation}\label{app:algorithms}

\paragraph{Forward and inverse maps.}
When applying the forward transformation to a singular sample correlation
matrix, a positive-definite estimator must first be obtained.
For the forward map, obtain $K=C^{-1}$ from a Cholesky factorization,
set $\delta_i=K_{ii}^{1/4}$, and evaluate
$B_{ij}=\tfrac12\{\delta_i\delta_jC_{ij}-K_{ij}/(\delta_i\delta_j)\}$.
Set the diagonal to zero to remove roundoff. For the inverse, compute
$B=Q\diag(\lambda_i)Q'$, evaluate $a_i=\exp\{\arsinh(\lambda_i)\}$,
and return $\Cor\{Q\diag(a_i)Q'\}$. This avoids cancellation in
$\lambda_i+\sqrt{1+\lambda_i^2}$ for large negative $\lambda_i$.

\paragraph{Derivative products.}
Reuse the inverse eigendecomposition in
\Cref{prop:differential-inverse} and \eqref{eq:VJP-BFT}. With
$u_i=\arsinh(\lambda_i)$, evaluate the divided differences as
\[
 \Omega_{ij}=2\operatorname{expit}(u_i+u_j),
 \qquad \operatorname{expit}(x)=(1+e^{-x})^{-1},
\]
using an overflow-safe logistic routine. The derivative and adjoint then
require matrix products and diagonal operations, without constructing
a $d\times d$ Jacobian. Map-only evaluations do not form $\Omega$.

\paragraph{LS scaling.}
For the LS forward map, solve $Cx=x^{-1}$ with $x>0$ and evaluate
$\log\{\diag(x)C\diag(x)\}$. This positive diagonal scaling to unit row sums
exists uniquely \citep{MarshallOlkin:1968,JohnsonReams:2009}.
The damped Newton step solves
$\{C+\diag(x^{-2})\}s=-(Cx-x^{-1})$, starting from $x=\one_n$ and
restricting the step to preserve positivity
\citep[Section~3.5]{Thanwerdas:2024}.

\paragraph{InvC inversion.}
The diagonal-completion formulation follows
\citet[Section~4.2]{AmendolaZwiernik:2021}; its interpretation as unrestricted
correlation coordinates is established by
\citet[Corollary~5.8]{Zwiernik:2025}.
For hollow precision coordinates $W$, set $K(v)=W+\diag(v)$ and
$C=K(v)^{-1}$. Recover the missing diagonal by minimizing
$\phi(v)=\one_n'v-\log\det K(v)$ over $K(v)\succ0$.
Its gradient is $\one_n-\diag(C)$ and its Hessian is $C\Had C\succ0$,
so the objective is strictly convex with first-order condition
$\diag(C)=\one_n$. With
$C^{(k)}=K(v^{(k)})^{-1}$, the Newton step solves
\[
 (C^{(k)}\Had C^{(k)})s^{(k)}
 =\diag(C^{(k)})-\one_n.
\]
Initialize $v_i^{(0)}=1+\sum_{j\ne i}|W_{ij}|$, which ensures positive
definiteness by strict diagonal dominance. Update
$v^{(k+1)}=v^{(k)}+\alpha_ks^{(k)}$, backtracking from
$\alpha_k=1$ to maintain positive definiteness and satisfy Armijo decrease
of $\phi$. Return $C^{(k)}$ when
$\|\diag(C^{(k)})-\one_n\|_\infty$ is below the chosen tolerance.

\section{Sparsity preservation: proofs and comparisons}
\label{app:sparsity-proofs}

This appendix supplies the proofs for Section~\ref{sec:high-dimensional-sparsity}
of the main paper, followed by calculations and counterexamples for GFT,
LS, and inverse-correlation (InvC) coordinates. The notation and factor assumptions are those of
\Cref{ass:sparse-factor}. No additional conditions are imposed unless stated.

\subsection{Proofs of the BCP results}

\begin{proof}[Proof of \Cref{thm:sparse-bft}]
Suppress the subscript $n$ and set
$M=I_{r_n}+\Lambda'\Omega^{-1}\Lambda$, so that
$M\succeq(1+\mu_n)I_{r_n}$.
Woodbury's identity gives
\begin{equation}\label{eq:sparse-woodbury}
 K=\Omega^{-1}-\Omega^{-1}\Lambda M^{-1}\Lambda'\Omega^{-1}.
\end{equation}
The loading bound and \eqref{eq:sparse-factor-bounds} imply
\begin{equation*}
 |K_{ij}|\leq h_n\quad(i\ne j),\qquad
 0\leq\omega_i^{-1}-K_{ii}\leq h_n.
\end{equation*}
For a correlation matrix, a Schur complement gives $K_{ii}\geq1$.
Equation \eqref{eq:sparse-woodbury} also gives
$K_{ii}\leq\omega_i^{-1}\leq a^{-1}$.
Since the derivative of $x^{1/4}$ is at most $1/4$ on $[1,\infty)$,
\begin{equation*}
 |\delta_i-\omega_i^{-1/4}|\leq h_n/4,\qquad
 |\delta_i\delta_j-(\omega_i\omega_j)^{-1/4}|
 \leq a^{-1/4}h_n/2.
\end{equation*}
Substitute these inequalities into \eqref{eq:sparse-bft-exact}, using
$|C_{ij}|\leq1$ and $\delta_i\delta_j\geq1$.
This proves \eqref{eq:sparse-bft-expansion}--\eqref{eq:sparse-bft-error}.
When $C_{ij}=0$, only the precision term remains, proving
\eqref{eq:sparse-null-bound}.
\end{proof}

\begin{proof}[Proof of \Cref{cor:sparse-support}]
Null coordinates satisfy \eqref{eq:sparse-null-bound}. On the support,
\Cref{thm:sparse-bft} gives
$|B_{ij,n}|\geq m_n/2-K_a h_n$, and the perturbation is smaller than
the positively scaled signal under the displayed conditions. The bounds
hold simultaneously over all pairs, proving support and sign recovery.
\end{proof}

\begin{proof}[Proof of \Cref{cor:sparse-estimated-support}]
Write $E_n=\max_{i\ne j}|\hat B_{ij,n}-c_nB_{ij,n}|$.
At null pairs, $|\hat B_{ij,n}|\leq c_nh_n/2+E_n=o_p(\tau_n)$.
At nonnull pairs,
$|\hat B_{ij,n}|\geq c_n(m_n/2-K_ah_n)-E_n$.
Since $c_n^{-1}=O_p(1)$, dividing the threshold and both error terms
by $c_nm_n$ gives quantities converging to zero in probability.
The lower bound therefore exceeds $\tau_n$ simultaneously over the
support with probability tending to one. The same bounds make the total
perturbation $o_p(c_nm_n)$, proving sign recovery because $c_n>0$.
\end{proof}

\begin{proof}[Proof of \Cref{prop:sparse-reconstruction}]
Suppress $n$, write $\mu=\mu_n$ and $r=r_n$, and use
$\norm{H}_{\max}=\max_{i,j}|H_{ij}|$.
Set $K=C^{-1}$, $\Delta=\diag(K_{ii}^{1/4})$, and
$W=\Omega^{-1}-K$. By \eqref{eq:sparse-woodbury},
$W\succeq0$, $\rank(W)\leq r$, $\norm{W}_2\leq a^{-1}$,
and $\max_i W_{ii}\leq h_n$.
Define
\begin{equation*}
 U=\frac{\Delta\Lambda}{\sqrt2},\qquad
 R=\frac12\Delta^{-1}W\Delta^{-1},\qquad
 D_0=\diag\!\left(\frac{q_i-q_i^{-1}}2\right),\qquad
 q_i=\omega_i\sqrt{K_{ii}}.
\end{equation*}
Since $1\leq K_{ii}\leq\omega_i^{-1}$, we have $a\leq q_i\leq1$,
so $D_0$ is uniformly bounded. Direct substitution gives
$B=D_0+UU'+R$. The loading and residual-variance bounds imply
\begin{equation*}
 \lambda_{\min}(U'U)\geq a\mu/2,\qquad
 \max_i\norm{U_{i\cdot}}_2^2\leq\frac{1-a}{2\sqrt a}.
\end{equation*}
Indeed, $U'U\succeq\Lambda'\Lambda/2
\succeq a\Lambda'\Omega^{-1}\Lambda/2$.
Moreover, $R\succeq0$, $\rank(R)\leq r$, and
\begin{equation*}
 \norm{R}_2\leq\frac1{2a},\qquad
 \norm{R}_{\max}\leq\max_iR_{ii}\leq\frac{h_n}{2},\qquad
 \max_i\norm{R_{i\cdot}}_2^2
 \leq\norm{R}_2\max_iR_{ii}\leq\frac{h_n}{4a}.
\end{equation*}

Let $T$ retain the diagonal of $R$ and its off-diagonal entries wherever
$C_{ij}\ne0$, setting the other entries to zero. At a correlation zero,
$(UU')_{ij}=0$, so $\tilde B=D_0+UU'+T$ exactly.
Masking entries cannot increase a row's Euclidean norm or the Frobenius
norm. Consequently, for either $J=R$ or $J=T$,
\begin{equation*}
 \norm{J}_{\max}=O(\mu^{-1}),\qquad
 \max_i\norm{J_{i\cdot}}_2=O(\mu^{-1/2}),\qquad
 \norm{J}_2=O(1),
\end{equation*}
where the last bound for $T$ follows from
$\norm{T}_2\leq\normF{T}\leq\normF{R}
\leq\sqrt r/(2a)$ and $\sup_n r_n<\infty$.

We next establish the entrywise bound needed for the inverse map.
Let $X=D_0+UU'+J$, $L=D_0+J$, $H=U'U$, and
$V=UH^{-1/2}$. Then $V'V=I_r$,
$X=L+VHV'$, and
$\max_i\norm{e_i'V}_2=O(\mu^{-1/2})$, where $e_i$ is the $i$th
coordinate vector. Choose a fixed $M\geq1$ with
$\norm{L}_2,\norm{D_0}_2\leq M$.
The min--max principle places $n-r$ eigenvalues of $X$ in $[-M,M]$
and its remaining $r$ eigenvalues in $[a\mu/2-M,\infty)$.
For sufficiently large $n$, the latter interval lies to the right of $M+1$.

Set $g(z)=\sqrt{1+z^2}-z$, taking the square-root branch that is
positive on the real axis; $g$ is analytic on $|\operatorname{Im}z|<1$.
Let $\Gamma$ be the positively oriented rectangle with vertical sides
$\operatorname{Re}z=\pm(M+1)$ and horizontal sides
$\operatorname{Im}z=\pm1/2$. On this fixed contour, the resolvents
\begin{equation*}
 \mathcal R(z)=(zI_n-L)^{-1},\qquad
 \mathcal R_0(z)=(zI_n-D_0)^{-1}
\end{equation*}
are uniformly bounded in spectral norm, and $\mathcal R_0(z)$ is
diagonal. The resolvent identities
\begin{align*}
 \mathcal R-\mathcal R_0
 &=\mathcal R_0J\mathcal R_0
   +\mathcal R_0J\mathcal R J\mathcal R_0,\\
 \mathcal RV&=\mathcal R_0V+\mathcal R_0J\mathcal RV
\end{align*}
and the entrywise and row-norm bounds above yield, uniformly on $\Gamma$,
\begin{equation*}
 \norm{\mathcal R-\mathcal R_0}_{\max}=O(\mu^{-1}),\qquad
 \max_i\norm{e_i'\mathcal RV}_2=O(\mu^{-1/2}).
\end{equation*}
For the quadratic resolvent term, each entry is bounded by a constant
times $\norm{J_{i\cdot}}_2\norm{J_{j\cdot}}_2$; for the second identity,
use $\norm{\mathcal RV}_2=O(1)$.

The compression $\mathcal K(z)=V'\mathcal R(z)V$ is uniformly invertible on
$\Gamma$. To see this, for every $\lambda\in[-M,M]$, the imaginary
part of $(z-\lambda)^{-1}$ has a common sign and magnitude bounded
away from zero on either horizontal side. The same is true of its real
part on either vertical side. The spectral theorem and $V'V=I_r$
therefore give $|x^*\mathcal K(z)x|\geq c\norm{x}_2^2$ for all complex $x$,
with a constant $c>0$ independent of $n$ and $z\in\Gamma$.
Since $\norm{H^{-1}}_2=O(\mu^{-1})$,
$(H^{-1}-\mathcal K(z))^{-1}$ is also uniformly bounded. Woodbury's identity gives
\begin{equation*}
 (zI_n-X)^{-1}-\mathcal R(z)
 =\mathcal R(z)V\{H^{-1}-\mathcal K(z)\}^{-1}V'\mathcal R(z).
\end{equation*}
The preceding row-norm bound, and the symmetry of the resolvent,
show that this difference is $O(\mu^{-1})$ in maximum-entry norm.

Let $\Pi$ be the spectral projection of $X$ onto the eigenvalues inside
$\Gamma$. Cauchy's integral formula now gives
\begin{equation*}
 g(X)\Pi-g(D_0)
 =\frac{1}{2\pi\mathrm{i}}\int_\Gamma
 g(z)\{(zI_n-X)^{-1}-\mathcal R_0(z)\}\,\mathrm dz,
\end{equation*}
whose maximum-entry norm is $O(\mu^{-1})$.
The remaining eigenvalues are at least $a\mu/2-M$, and
$g(t)=1/(\sqrt{1+t^2}+t)=O(t^{-1})$ as $t\to\infty$.
Thus $\norm{g(X)(I_n-\Pi)}_2=O(\mu^{-1})$, proving
\begin{equation*}
 \norm{g(X)-g(D_0)}_{\max}=O(\mu^{-1})
\end{equation*}
for both $X=B$ and $X=\tilde B$.

For $f(x)=x+\sqrt{1+x^2}=2x+g(x)$, it follows that
\begin{equation*}
 \norm{f(\tilde B)-f(B)}_{\max}=O(\mu^{-1}),
\end{equation*}
because $\norm{T-R}_{\max}=O(\mu^{-1})$.
Originally $A=f(B)=\Delta C\Delta$ satisfies
$1\leq A_{ii}\leq a^{-1/2}$ and
$|A_{ij}|\leq\sqrt{A_{ii}A_{jj}}\leq a^{-1/2}$.
The same diagonal bounds, up to an $O(\mu^{-1})$ perturbation, hold
for $\tilde A=f(\tilde B)$.
Normalizing each entry by the square root of its two diagonal entries
therefore preserves the uniform $O(\mu^{-1})$ error.
This proves the proposition, including the $O(n^{-1})$ rate under
pervasive factors.
\end{proof}

\paragraph{Extension to an increasing number of factors.}
The proof uses the bound on $r_n$ only to obtain $\norm{T_n}_2=O(1)$.
All subsequent constants are independent of $r_n$. Thus the same
reconstruction bound holds under \Cref{ass:sparse-factor} with an
increasing number of factors whenever the masked precision correction
$T_n$, defined in the proof, has uniformly bounded spectral norm.
This is a sufficient condition; entrywise smallness alone is not used
to infer a spectral-norm bound.

\subsection{Gaussian plug-in coordinates}
\label{app:sparse-plugin}

\begin{proposition}[Common scaling in Gaussian plug-in coordinates]
\label{prop:sparse-plugin}
Let $\hat C_n$ be the sample correlation matrix of $T$ independent,
identically distributed Gaussian vectors with positive-definite covariance
matrix $\Sigma_n$. Suppose its population correlation matrix satisfies
$\lambda_{\min}(C_n)\geq a>0$ uniformly in $n$.
As $n,T\rightarrow\infty$, assume $n/T\leq\bar\gamma<1$ and
$\log n=o(T)$. With $c_n=(1-n/T)^{-1/2}$,
\begin{equation*}
 \max_{i\ne j}\bigl|[B(\hat C_n)]_{ij}-c_n[B(C_n)]_{ij}\bigr|
 =O_p\bigl(\sqrt{\log n/T}\bigr).
\end{equation*}
No factor representation of $C_n$ is required.
\end{proposition}

\begin{proof}
Correlation is invariant to positive diagonal rescaling, so take
$\Sigma_n=C_n$. Suppress $n$, set $K=C^{-1}$ and
$e_n=\sqrt{\log n/T}$, and let $S$ be the centered sample covariance
with divisor $T$. Then $TS\sim W_n(C,T-1)$ and is positive definite
almost surely eventually. Gaussian concentration and a union bound give
\begin{equation*}
 \max_i|S_{ii}-1|=O_p(e_n),\qquad
 \max_{i\ne j}|\hat C_{ij}-C_{ij}|=O_p(e_n).
\end{equation*}
For each pair $J=\{i,j\}$, the principal block
$[(TS)^{-1}]_{JJ}$ is the inverse of a two-dimensional Wishart matrix
with $\nu=T-n+1$ degrees of freedom and scale $K_{JJ}^{-1}$
\citep[Section~3.2]{Muirhead:1982}. Thus its marginal distribution is that of
\begin{equation*}
 K_{JJ}^{1/2}(Z_J'Z_J)^{-1}K_{JJ}^{1/2},
\end{equation*}
where $Z_J$ has $\nu$ rows of independent standard bivariate normals.
Only these marginal tail bounds are needed; matrices for different
pairs need not be independent. Fixed-dimensional Wishart concentration,
inversion on a neighborhood of $I_2$, and a union bound over all pairs give
\begin{equation*}
 \max_J\norm{(Z_J'Z_J/\nu)^{-1}-I_2}_2
 =O_p\bigl(\sqrt{\log n/\nu}\bigr).
\end{equation*}
Here the maximum expresses the uniform tail bound under the joint
sample law. Since $\norm{K_{JJ}}_2\leq a^{-1}$,
$\nu\geq(1-\bar\gamma)T$, and $T/\nu=c_n^2+O(T^{-1})$, it follows that
\begin{equation*}
 \max_{i,j}|(S^{-1})_{ij}-c_n^2K_{ij}|=O_p(e_n).
\end{equation*}
Also $\hat K=\hat C^{-1}
=D_S^{1/2}S^{-1}D_S^{1/2}$, where $D_S=\Diag(S)$, so the same bound
holds with $S^{-1}$ replaced by $\hat K$.
Finally write the exact BCP entry as
\begin{equation*}
 g(x,p,q,r)=\tfrac12\{(pq)^{1/4}x-(pq)^{-1/4}r\}.
\end{equation*}
It satisfies $g(x,c^2p,c^2q,c^2r)=c\,g(x,p,q,r)$ and is Lipschitz on
$p,q\in[1,L]$, $|x|\leq1$, $|r|\leq L$ for each fixed $L$.
Precision diagonals of correlation matrices are at least one, while
$c_n^2\leq(1-\bar\gamma)^{-1}$ and $\norm{K}_2\leq a^{-1}$.
The sample and scaled population arguments therefore lie in such a
fixed set with probability tending to one. Combining the uniform bounds
and homogeneity proves the result.
\end{proof}

When $n/T$ has a positive limit, the common scale can differ persistently
from one. This need not prevent support recovery: under the factor
assumption, \Cref{cor:sparse-estimated-support} applies with
$e_n=\sqrt{\log n/T}$ when its threshold conditions hold.
The proposition concerns independent Gaussian sampling; it does not
supply a universal correction for dependent or heavy-tailed observations.

\subsection{GFT and InvC bounds and counterexample}
\label{app:sparse-ah}

Let $H_n=\Off(\log C_n)$, following
\citet{ArchakovHansen:Correlation}. The following lemma underlies the GFT and LS shrinkage bounds.

\begin{lemma}[Entrywise logarithm bound for a diagonal-plus-factor matrix]
\label{lem:sparse-log-bound}
Let $V_n=D_n+U_nU_n'$, where $U_n\in\mathbb R^{n\times r_n}$ and $D_n$ is diagonal with entries in a fixed
interval $[\omega_-,\omega_+]\subset(0,\infty)$ independent of $n$. Suppose the rows $u_i'$ of
$U_n$ have norm at most $L$, and
$g_n:=\lambda_{\min}(U_n'U_n)>0$.
Then
\begin{equation*}
 \max_{i,j}|[\log V_n-\log D_n]_{ij}|
 =O(\log(1+g_n)/g_n).
\end{equation*}
\end{lemma}

\begin{proof}
The integral representation of the matrix logarithm and Woodbury's
identity give, with $R_t=(D_n+tI_n)^{-1}$,
\begin{equation}\label{eq:sparse-log-integral}
 \log V_n-\log D_n
 =\int_0^\infty R_tU_n
 (I_{r_n}+U_n'R_tU_n)^{-1}U_n'R_t\,\dd t.
\end{equation}
For every pair, the absolute value of the integrand is at most
\begin{equation*}
 \frac{L^2}{(\omega_-+t)^2}
 \frac{\omega_++t}{\omega_++t+g_n}
 \leq\frac{c_1}{(1+t)(1+t+g_n)},
\end{equation*}
where $c_1$ depends only on $\omega_-,\omega_+,L$.
Integration yields $c_1\log(1+g_n)/g_n$, with absolute convergence
entrywise.
\end{proof}

\begin{proposition}[Uniform GFT shrinkage]\label{prop:sparse-ah-shrinkage}
Under \Cref{ass:sparse-factor},
\begin{equation*}
 \max_{i\ne j}|H_{ij,n}|=O(\log(1+\mu_n)/\mu_n).
\end{equation*}
More explicitly,
\begin{equation*}
 \max_{i\ne j}|H_{ij,n}|
 \leq\frac{1-a}{a}\frac{\log(1+\mu_n)}{\mu_n}.
\end{equation*}
These bounds apply to all pairs, including correlations bounded away from zero.
\end{proposition}

\begin{proof}
Since $\Omega_n^{-1}\preceq a^{-1}I_n$,
$\Lambda_n'\Lambda_n\succeq a\mu_n I_{r_n}$.
Apply \Cref{lem:sparse-log-bound} with $D_n=\Omega_n$, $U_n=\Lambda_n$,
and $L=1$; $\log\Omega_n$ is diagonal.
For the explicit bound, set $R_t=(\Omega_n+tI_n)^{-1}$. Since
$R_t\succeq\{a/(a+t)\}\Omega_n^{-1}$ and
$\norm{\lambda_i}^2\leq1-a$, the integrand in
\eqref{eq:sparse-log-integral} has entries bounded in absolute value by
\[
 \frac{1-a}{(a+t)^2}
 \frac{a+t}{a+t+a\mu_n}.
\]
Its integral is $(1-a)\log(1+\mu_n)/(a\mu_n)$.
\end{proof}

Vanishing of every coordinate does not by itself preclude support
recovery, since a common rescaling leaves their ordering unchanged. The
failure in general is that the leading GFT term can have a different zero
pattern.

\begin{proposition}[Leading GFT term with homogeneous residuals]
\label{prop:sparse-ah-leading}
Under \Cref{ass:sparse-factor}, suppose $\Omega_n=\omega I_n$ with
fixed $0<\omega<1$ and $\mu_n\geq\kappa n$. Set
$G_n=\Lambda_n'\Lambda_n$ and $Q_n=G_n/n$. Uniformly for $i\ne j$,
\begin{align}
 H_{ij,n}
 &=\lambda_{i,n}'G_n^{-1}\log(I_{r_n}+G_n/\omega)\lambda_{j,n}
 \label{eq:sparse-ah-exact}\\
 &=\frac{\log n}{n}\lambda_{i,n}'Q_n^{-1}\lambda_{j,n}+O(n^{-1}).
 \label{eq:sparse-ah-leading}
\end{align}
\end{proposition}

\begin{proof}
Spectral calculus gives
$\log C_n=(\log\omega)I_n+
\Lambda_nG_n^{-1}\log(I_{r_n}+G_n/\omega)\Lambda_n'$.
The eigenvalues of $Q_n$ are bounded above by $\tr(G_n)/n\leq1$
and below by $\omega\kappa$. Thus
$\log(I_{r_n}+nQ_n/\omega)=(\log n)I_{r_n}+O(1)$ in operator norm,
which proves \eqref{eq:sparse-ah-leading} using bounded loading rows.
\end{proof}

\begin{example}[Failure of common-threshold separation for GFT and InvC]
\label{ex:sparse-ah-failure}
Let $n=4m+3$, $\Omega_n=\tfrac12I_n$, and take the first three loading
rows as
\[
 \lambda_1=2^{-1/2}(1,0)',\quad
 \lambda_2=2^{-1/2}(0,1)',\quad
 \lambda_3=10^{-1/2}(1,2)'.
\]
Add $3m$ rows $(1/2,1/2)$ and $m$ rows $(1/2,-1/2)$. All row norms
squared equal $1/2$, and
\[
 G_n=nQ+O(1),\qquad
 Q=\begin{pmatrix}1/4&1/8\\1/8&1/4\end{pmatrix},\qquad
 Q^{-1}=\begin{pmatrix}16/3&-8/3\\-8/3&16/3\end{pmatrix}.
\]
The eigenvalues of $Q$ are $3/8$ and $1/8$, so the factor assumptions
hold. Direct calculation gives $C_{12,n}=0$ and
$C_{13,n}=1/(2\sqrt5)$, while \eqref{eq:sparse-ah-exact} and
Woodbury's identity give
\[
 \begin{aligned}
 H_{12,n}&\sim-\frac43\frac{\log n}{n},&
 nH_{13,n}&\rightarrow\frac{2\log3}{\sqrt5},\\
 K_{12,n}&=\frac8{3n}+O(n^{-2}),&
 K_{13,n}&=O(n^{-2}).
 \end{aligned}
\]
Indeed, $\lambda_1'Q^{-1}\lambda_2=-4/3$ and
$\lambda_1'Q^{-1}\lambda_3=0$. The next GFT term is
$\lambda_1'Q^{-1}\log(2Q)\lambda_3/n
=2\log3/(n\sqrt5)$, and
$K_{ij,n}=-2\lambda_i'Q^{-1}\lambda_j/n+O(n^{-2})$.
In both coordinate systems, a common magnitude threshold retaining the
nonnull pair $(1,3)$ therefore also retains the null pair $(1,2)$,
even after common rescaling. The finitely many loading types have a
fixed positive minimum nonzero correlation.
\end{example}

For completeness, InvC's $O(\mu_n^{-1})$ shrinkage follows directly from
\eqref{eq:sparse-woodbury}.

\subsection{LS bounds and counterexample}
\label{app:sparse-ls}

Use $S_n$ and $D_n$ defined in \Cref{sec:sparsity-comparators} of the main
paper. Equivalently, $D_n=\diag(x)$ with $C_nx=x^{-1}$ componentwise.
The transformed matrix has zero row sums, and its inverse is $\Cor(e^{S_n})$.

\begin{proposition}[LS shrinkage with comparable scaling factors]
\label{prop:sparse-ls-regular}
Suppose \Cref{ass:sparse-factor} holds and
\begin{equation}\label{eq:sparse-ls-scale-condition}
 \sup_n\frac{\max_i(D_n)_{ii}}{\min_i(D_n)_{ii}}<\infty.
\end{equation}
Then $\max_{i\ne j}|S_{ij,n}|=O(\log(1+\mu_n)/\mu_n)$, including at fixed
nonzero correlations.
\end{proposition}

\begin{proof}
Let $x_{\min}=\min_ix_i$ and $\tilde D_n=D_n/x_{\min}$.
The entries of $\tilde D_n$ lie in $[1,L]$ for a fixed $L$.
The matrix
\begin{equation*}
 \tilde D_n C_n\tilde D_n
 =\tilde D_n\Omega_n\tilde D_n+
 (\tilde D_n\Lambda_n)(\tilde D_n\Lambda_n)'
\end{equation*}
satisfies \Cref{lem:sparse-log-bound}, with residual diagonal in
$[a,L^2]$, bounded loading rows, and loading Gram matrix at least
$\Lambda_n'\Lambda_n\succeq a\mu_nI_{r_n}$.
The scalar rescaling contributes only $2(\log x_{\min})I_n$ to the logarithm,
which has no off-diagonal entries.
\end{proof}

Condition \eqref{eq:sparse-ls-scale-condition} is sufficient, not
necessary. It is not implied by \Cref{ass:sparse-factor}, and the next
example shows why no uniform LS sparsity guarantee follows from the
original assumptions alone.

\begin{example}[LS fills a null pair and shrinks a nonnull pair]
\label{ex:sparse-ls-failure}
Let $n=8m+4$, $\Omega_n=\tfrac12I_n$, and $s=1/\sqrt2$.
Take the first four loading rows as
$(0,s,0),(0,0,s),(0,-s,0),(0,0,-s)$, followed by $3m$ copies each of
$(1/2,s/2,s/2)$ and $(1/2,-s/2,-s/2)$, and $m$ copies each of
$(1/2,s/2,-s/2)$ and $(1/2,-s/2,s/2)$.
All row norms squared equal $1/2$. The loading Gram eigenvalues are
$2m$, $1+3m/2$, and $1+m/2$, so the factor assumptions hold.

Here $C_{12,n}=0$ and direct substitution into
$D_nC_nD_n\one_n=\one_n$ gives
\[
 D_n=\diag\{\sqrt2 I_4,(2m+1/2)^{-1/2}I_{8m}\}.
\]
Set $\epsilon_m=(4m+1)^{-1}$,
$t_{+,m}=3m/(4m+1)$, and $t_{-,m}=m/(4m+1)$.
Rotating the special-variable contrasts
$(e_1-e_3)/\sqrt2,(e_2-e_4)/\sqrt2$ to their sum and difference,
and pairing them with the corresponding normalized bulk contrasts,
reduces the relevant scaled covariance to
\[
 K_{\pm,m}=\begin{pmatrix}
 3&\sqrt{2t_{\pm,m}}\\
 \sqrt{2t_{\pm,m}}&t_{\pm,m}+\epsilon_m
 \end{pmatrix}.
\]
The special-variable sums have eigenvalue one. Consequently, with
$f(t)=[\log\left(\begin{smallmatrix}3&\sqrt{2t}\\\sqrt{2t}&t\end{smallmatrix}\right)]_{11}$,
\[
 S_{12,n}=\tfrac14\{[\log K_{+,m}]_{11}-[\log K_{-,m}]_{11}\}
 \rightarrow\tfrac14\{f(3/4)-f(1/4)\}<0.
\]
The strict inequality follows analytically from the resolvent derivative
\[
 f'(t)=-\int_0^\infty
 \frac{2z}{\{t(1+z)+z(3+z)\}^2}\,\dd z<0.
\]
For two distinct variables in the first bulk group, $C_{56,n}=1/2$.
Let $J$ contain the eight normalized group indicators and
$A_8=J'D_nC_nD_nJ$. Block functional calculus gives
$S_{56,n}=\{[\log A_8]_{55}-\log\epsilon_m\}/(3m)=O(\log n/n)$:
the scaled residual variances are at least $\epsilon_m$, and both the
scaled residual diagonal and loading Gram matrix are uniformly bounded.
Thus $S_{56,n}\rightarrow0$ while the null-pair coordinate
has a nonzero limit. No common magnitude threshold separates the two,
although the finitely many loading types have nonzero correlations
bounded away from zero.
\end{example}

\subsection{Finite-dimensional qualifications}
\label{app:sparse-qualifications}

\begin{remark}[Dimension alone is insufficient]\label{rem:sparse-exact-zeros}
By \eqref{eq:sparse-bft-exact}, at a correlation zero $B_{ij}=0$ if and
only if $K_{ij}=0$. The positive-definite sequence
\[
 C_n=\begin{pmatrix}1&0&1/2\\0&1&1/2\\1/2&1/2&1\end{pmatrix}
 \oplus I_{n-3}
\]
has $C_{12,n}=0$ but $B_{12,n}=-1/(2\sqrt6)$ for every $n$.
Its smallest eigenvalue is $1-1/\sqrt2>0$; increasing dimension and
a lower eigenvalue bound alone therefore do not ensure preservation.
\end{remark}

The partial-correlation row of \Cref{tab:method-comparison} concerns
marginal support. To see the distinction, take loading rows
$(1/\sqrt2,0),(0,1/\sqrt2),(1/2,1/2),(\eta,0)$ with
$0<\eta<2/7$, and residual variances giving unit diagonals.
Then $C_{12}=0$ but $\rho_{12\mid3}=-1/7$, whereas $C_{34}=\eta/2<1/7$.
In a sequential parametrization with variable 3 first, both are
coordinates. Adding positive proportions of the two coordinate loading
directions makes the factors pervasive without altering these values.
No common threshold, before or after Fisher's scalar transformation,
retains the latter coordinate while excluding the former.

\section{Empirical implementation and inverse diagnostics}
\label{app:empirical-implementation}

\paragraph{Data and sample construction.}
We match CRSP daily returns for all available tickers in 2011--2025 to
Compustat eight-digit GICS subindustry codes (\texttt{gsubind}), retaining
all matched observations. Securities are identified by PERMNO, without
issuer aggregation or additional exchange, share-type, capitalization,
price, volume, or index-membership screens. The return is CRSP's
\texttt{DlyRetx}, a price return excluding dividends, in decimal units;
the risk-free rate is not subtracted. The FF5 and momentum predictors
are matched by date and converted from percentages to decimals;
\texttt{Mkt-RF} is the market predictor.

Of 8,670 securities, 1,765 have finite returns on all 3,772 dates from
January 3, 2011 to December 31, 2025. No returns are imputed or winsorized.
The 911 duplicate PERMNO--date records in the candidate balanced panel
are collapsed after checking agreement of returns, factors, and
classification fields; no conflicts occur. The figures use final
available GICS subindustry codes, with ticker and PERMNO breaking ties.

\paragraph{Volatility standardization and inference.}
Each full-period FF5-plus-momentum OLS regression includes an intercept.
The zero-mean GJR-GARCH(1,1) fit uses Gaussian quasi-maximum likelihood,
residuals multiplied by 100 for numerical scaling, and an exponentially
weighted backcast with decay 0.94 over the first 75 observations.
Variances are subsequently restored to return units. The restrictions are
$\omega_i>0$, $\alpha_i\geq0$, $\alpha_i+\gamma_i\geq0$, $\beta_i\geq0$,
and $\alpha_i+\gamma_i/2+\beta_i\leq1-10^{-6}$. Fitted residuals and variances are denoted $u_{it}$
and $h_{it}$, and $z_{it}=u_{it}/\sqrt{h_{it}}$. For the Gaussian working scores, the shocks are centered and scaled
using second moments with divisor $T$. Bartlett HAC tests use nine lags
and the same normalization, treating OLS and volatility filters as fixed;
first-stage parameter uncertainty is not propagated. BIC counts the
retained correlation coordinates, with those filters common to all columns.

\paragraph{Restricted inverse maps.}
The common 5\% mask leaves retained off-diagonal coordinates unchanged;
no maximum-likelihood refit is performed. BCP uses its explicit inverse.
For restricted GFT coordinates $H$, the diagonal iteration is
$v\leftarrow v-\log\{\diag[\exp(H+\diag(v))]\}$, initialized at the
diagonal of $\log\hat C$. It converges in 64 iterations to a maximum
absolute log-diagonal error of $9.45\times10^{-13}$.
For LS, reset the diagonal to minus the off-diagonal row sums before
exponentiation and correlation normalization.
The InvC Newton scheme in \Cref{app:algorithms} attains a maximum
unit-diagonal error of $5.84\times10^{-14}$.
All four reconstructed matrices are positive definite, and the maximum
absolute coordinate recovery error is below $4\times10^{-12}$.

\end{document}